\documentclass{LMCS}

\def\dOi{10(1:1)2014}
\lmcsheading%
{\dOi}
{1--39}
{}
{}
{Mar.~12, 2011}
{Jan.~\phantom02, 2014}
{}

\ACMCCS{[{\bf Theory of computation}]: Semantics and
  reasoning---Program reasoning---Program verification\,/\,Program
  specifications; Semantics and reasoning---Program semantics; Formal
  languages and automata theory}


\usepackage{hyperref}
\usepackage{color}

\usepackage{url}
\usepackage[latin1]{inputenc}
\usepackage{amsmath}
\usepackage{amscd}
\usepackage{latexsym}
\usepackage{stmaryrd}
\usepackage{amssymb}
\usepackage{amscd}
\usepackage{fancyvrb}
\usepackage{boxedminipage}
\usepackage{listings}
\usepackage{ifpdf}
\usepackage{casl}
\usepackage{rotating}
\usepackage{textcomp}

\newcommand{\rotxc}[1]{\begin{sideways}#1\end{sideways}}
\newcommand{\invert}[1]{\rotxc{\rotxc{#1}}}
\newcommand{\multiarr}{\rotxc{\invert{$\curlyveeuparrow$}}}
\newcommand{\nt}{~}
\newcommand{\Alg}{\mathrm{Alg}}

\long\def\CUT#1{\marginpar{CUT}}
\long\def\CUT#1{\relax}

\newenvironment{tra}[2]{translation from $\Sigma$ to $\Sigma'$ w.r.t. the set of variables
$#1$ and $#2$}

\def\fuc#1{\mathsf{#1}}

\def\imp{\mathbin{\Rightarrow}}
\def\impp{\mathbin{\rightarrow}}

\def\e{\mathbin{\wedge}}
\def\ou{\mathbin{\vee}}

\def\ie{{\em i.e.\/}}

\def\eg{{\em e.g.\/}}
\def\balance{ \mathsf{bal}}
\def\dep{ \mathsf{deposit}}
\def\withd{ \mathsf{withdraw}}
\def\max{ \mathsf{max}}
\def\transvalid{ \mathsf{val}}

\def\perm{ \mathsf{conf}}

\newcommand{\Kal}{K}
\newcommand{\ecr}{\models_{\Kal}}
\newcommand{\Va}{\mathrm{VAR}}

\newcommand{\theo}{\mathrm{Thm}}
\newcommand{\Ceq}{\mathrm{Ceq}}
\newcommand{\et}{t \approx t'}
\newcommand{\ceq}{t_1\approx t'_1 \wedge \cdots \wedge t_n \approx t'_n \rightarrow t \approx
t'}

\newcommand{\cM}{\mathcal{M}}

\newenvironment{form}[2]{\Fm^{#1}(#2)}

\newenvironment{spm}[1]{\mean{#1}}

\newcommand{\taupe}{\tau_{\cl,\Kal}}

\newcommand{\thy}{\mathrm{Th}}

\def\imp{\mathbin{\Rightarrow}}
\def\impp{\mathbin{\rightarrow}}

\long\def\CUT#1{\marginpar{CUT}} \long\def\CUT#1{\relax}

\def\caixapeq#1{\medskip
  \begin{center}
  \fbox{\begin{minipage}{0.75\textwidth}\protect{#1}\end{minipage}}
  \end{center}}

\def\caixapeqg#1{\medskip
  \begin{center}
  \fbox{\begin{minipage}{0.8\textwidth}\protect{#1}\end{minipage}}
  \end{center}}

\newcommand{\Th}{\mathrm{Th}}
\newcommand{\Thm}{\mathrm{Thm}}
\newcommand{\Mod}{\mathrm{Mod}}

\newcommand{\Con}{\mathrm{Cn}}

\newcommand{\AX}{\mathrm{AX}}

\newcommand{\CPC}{\mathrm{CPC}}
\newcommand{\HA}{\mathrm{HA}}
\newcommand{\BA}{\mathrm{BA}}

\newcommand{\IR}{\mathrm{IR}}

\newcommand{\Cn}{\mathrm{Cn}}

\newcommand{\Te}{\mathrm{Te}}

\newcommand{\cl}{\mathcal{L}}
\newcommand{\ca}{\mathcal{A}}

\newcommand{\mc}{\mathcal}

\newcommand{\Fm}{\mathrm{Fm}}
\newcommand{\kfm}{\Te^k_\Sigma(X)} % conjunto de k-formulas ou conjunto k-termos

\newcommand{\true}{\mathit{true}}
\newcommand{\false}{\mathit{false}}

\newcommand{\Eq}{\mathrm{Eq}}
\newcommand{\Sig}{\mathrm{Sig}}
\newcommand{\FV}{\mathrm{FV}}
\newcommand{\Vav}{\mathrm{VAR}}

\newcounter{romcount}

\def\dlogic{deductive system}
\def\dlogics{deductive systems}
\def\fuc#1{\mathsf{#1}}

\def\imp{\mathbin{\Rightarrow}}
\def\impp{\mathbin{\rightarrow}}

\def\e{\mathbin{\wedge}}
\def\ou{\mathbin{\vee}}
\def\true{\fuc{tt}}
\def\false{\fuc{ff}}
\def\ie{{\em i.e.\/}}

\def\eg{{\em e.g.\/}}

\newcommand{\eqq}[1]{\Eq({\Sigma_{#1})}}
\newcommand{\fdec}[3]{#1: #2 \longrightarrow #3}
\newcommand{\mfdec}[3]{#1: #2\,  \mathbin{\multiarr}\, #3}
\def\mean#1{\mathopen{[\![}#1\mathclose{]\!]}}
\def\setcat{{\sf Set}}
\def\tuple#1{\langle #1\rangle}
\def\pair#1{\tuple{#1}}

\def\balance{ \mathsf{bal}}
\def\bal{ \mathsf{bal}}
\def\dep{ \mathsf{deposit}}
\def\withd{ \mathsf{withdraw}}
\def\max{ \mathsf{max}}
\def\transvalid{ \mathsf{valid}}

\def\perm{ \mathsf{cf}}

\begin{document}

\title[The role of logical interpretations in program development]{The role of logical interpretations \\ in program development}

\author[M.~A.~Martins]{Manuel A. Martins\rsuper a}
\address{{\lsuper a}CIDMA - Center for R\&D in Mathematics and Applications, Dep. Mathematics, Univ. Aveiro, Aveiro, Portugal}
\email{martins@ua.pt}
% \thanks{thanks 1}

\author[A.~Madeira]{Alexandre Madeira\rsuper b}
\address{{\lsuper b}HASLab - INESC TEC,  Univ. Minho \& Dep. Mathematics, Univ. Aveiro, Portugal}
\email{madeira@ua.pt}
% \thanks{thanks 1}

\author[L.~S.~Barbosa]{Lu\'{\i}s S. Barbosa\rsuper c}
\address{{\lsuper c}HASLab - INESC TEC, Univ. Minho, Braga, Portugal}
\email{lsb@di.uminho.pt}
% \thanks{thanks 1}


\begin{abstract}
\noindent Stepwise refinement of algebraic specifications is a well known formal methodology for program development.  However, traditional notions of refinement based on signature morphisms are often too rigid to capture a number of relevant transformations in the context of software design, reuse, and adaptation. This paper proposes a new approach to
refinement in which signature morphisms are replaced by \emph{logical interpretations} as a means to witness refinements. The approach is first presented in the context of equational logic, and later generalised to deductive systems of arbitrary dimension. This allows, for example,  refining sentential into equational specifications and the latter into modal ones.
\end{abstract}

\maketitle

\section{Introduction}\label{sc:in}
%%%%%%%%%%%%%%%%
%\input{mmb12INTRO.tex}

\subsection{Context.}
The industrial demand for high-assurance software
 opens a window of opportunity for mathematically based development methods,
 able to design complex systems at ever-increasing levels of reliability and security.

 This paper's contribution is placed at a specific corner of the broad landscape of
 formal methods for software development: that of \emph{algebraic specification}
  \cite{EM85,Wir90,ST97,AKK99}, a family of methods which, having played a pioneering role,
  constitutes  at present a large and mature body of knowledge and active
 research.

  Such methods have a double origin. On the one hand
 they can be traced back to early work on data abstraction and
 modular decomposition of programs \cite{Par72, Hoa72,LZ74,Gut75,GH78}. On the
 other hand, to research on semantics of program specifications building on results from
 algebraic logic and model theory. Especially relevant in this respect is the original work
 of the so-called ADJ group \cite{GTWW77,GTW78} whose initial algebra semantics was the first,
 full formal approach to software development put forward. This double origin, temporally located
  around mid seventies, is not surprising: \emph{compositionality} is both
  a basic requirement in program development and a major asset in algebraic semantics.

 The whole area flourished rapidly from the outset: not only different approaches to semantics (final, observational,
 loose) emerged, but also the initial tie to many-sorted equational logic was soon extended,
 first to conditional-equational logic, and later to order-sorted, partial and  full first-order among other variants.
 The emergence of the first effective algebraic specification languages ---
 \textsc{Obj} \cite{GWMFJ96} and \textsc{Clear} \cite{CLEAR79} --- overlaps another major development:
 the introduction of institutions by J. Goguen and R. Burstall \cite{instituicoes}. Institution theory, which develops
model theory independent of the underlying logical system, made possible to decouple
 specification methodology from the particularities of whatever semantics one may consider
 more suitable to a specific problem \cite{Dia08}.

  Moreover, although for a long time the impact of these methods in industry has been limited, a
  successful effort has been made in the last 15 years towards convergence on generic
  frameworks with suitable tool support. The \textsc{Compass} and, later,
  the \textsc{CoFI} initiative \cite{COFI02}, which lead to the development of   \textsc{Casl} \cite{MHST03},
  are relevant milestones in this process. Besides \textsc{Casl}, Cafe\textsc{OBJ} \cite{cafeobjref}
  and \textsc{Maude} \cite{mauderef} are currently used in several industrial applications and
  tool development.
  Actually,
 research in such methods, either at a foundational or methodological level, found applications
 in new, unsuspected areas --- for example, in documenting service interfaces  \cite{HRD08} , characterising contracts in
 contract-based programming \cite{BH08} or test generation for software composition \cite{YKZZ08}.

\subsection{Motivation.}

  For the working software engineer, a software component
 is documented by an \emph{interface}, which provides a language through which
 it interacts  with its environment, and
 a \emph{specification} of the intended meaning of the services provided. This specification
 is implemented  by a concrete piece of software respecting the specified semantics.

Algebraic specification methods build on the observation that these somehow vague concepts
from Software Engineering can be framed rigorously in terms of well-known mathematical
notions. Thus, an \emph{interface} corresponds to a \emph{signature}, i.e., a set of names for the relevant types, called
\emph{sorts}, and a family of service or operation names, classified by their arity and input-output
sorts. A signature generates a formal language, giving a rigorous meaning to what we have called before
the component's interaction language. Once fixed a signature, a \emph{specification} describes
a \emph{class of models} for that signature, and an \emph{implementation} identifies a specific model
within such a class.
If functions provide suitable abstractions of the services offered by a software component,
this analogy can be made even more concrete by identifying  \emph{interfaces} with
\emph{algebraic signatures}, (denotations of) \emph{specifications} with \emph{classes of algebras}, and
\emph{implementations} with specific \emph{algebras}.

The analogy extends to the entire software development process along which components are \emph{refined}
 by incrementally adding detail and reducing under-specification.
Formally, this is a process of structural transformations witnessed by
\emph{signature morphisms}, which map functionally sorts and operations from a signature
to another respecting the sort translation of functional types.

In such a context this paper raises and discusses the following question: \emph{can more flexible
notions of refinement emerge from replacing signature morphisms  by some  weaker notion of transformation?}

The quest for \emph{weaker notions of transformation} lead us to a different setting, that of  Algebraic Logic \cite{FJP_survey}.
The key conceptual tool is that of a \emph{deductive system},  i.e., a formal language generated by a \emph{signature},
and a \emph{consequence relation}.
Interrelating such systems, through maps connecting logical properties, has been studied from early in the last century.
Such maps were called \emph{translations} and investigated as part of an ambitious programme
addressing tools to handle the multiplicity of logics.
As a result, several intuitive notions of translation are scattered in the literature.
Many logicians tailored the notion, for their own purposes, to relate specific logics and to
obtain specific results. In general, however, a translation is regarded as a map between
sets of formulas of different logics such that the image of a theorem is still a theorem.
They were used originally to clarify the relationship between classical and constructive logics.

Our starting point is the observation that specifications describe classes of models and those can be naturally associated to
deductive systems. Then, translations that  both reflect and preserve consequence relations seem  interesting
candidates to witness weaker forms of refinement. In this paper we will single out a specific sort of  translations
based on \emph{multifunctions},  \ie,  functions mapping an element to a set of
elements.  Such translations are called \emph{interpretations} and constitute a central tool in the study of
equivalent algebraic semantics (see, \eg, \cite{Woj,memoirs,pigozzi,blok,proto}).
A paradigmatic example is the interpretation of the \emph{classical propositional calculus} into the
\emph{equational theory of Boolean algebras} (cf. \cite[Example 4.1.2]{pigozzi}).
This paper explores interpretations
between the deductive systems corresponding to classes of models of specifications
as possible witnesses of refinement steps.
The notion seems able to capture a number of transformations which are difficult to deal with in classical terms.
Examples include data encapsulation and the decomposition of
operations into atomic transactions. It  also seems promising in the context of new, emerging computing paradigms which
entail the need for more flexible approaches to what counts as a valid transformation
along the development process  (see, for example, \cite{batory}).

\subsection{Contribution.}
 In this context, the contribution of this paper, which combines and extends previous results by the authors reported
 in \cite{MMB09} and \cite{MMB09h}, is twofold. On the one hand it puts forward a detailed  characterisation of
refinement witnessed by interpretations, referred in the sequel as \emph{refinement by interpretation},
 and exemplifies its potentialities in a number of small, yet illustrative examples.

 On the other hand, it renders the whole approach at a sufficiently abstract setting to be applicable
 over  \emph{arbitrary}, technically \emph{$k$-dimensional}, deductive systems. The dimension
 fixes the kind of relationship between terms one is interested in. Dimension 2, for example, encompasses
 equations, regarded as  instances of a  binary predicate asserting, for example, term equality, bisimilarity,
 or observational equivalence. Similarly, a unary predicate asserting the validity of a formula is enough to represent a proposition,
 leading to 1-dimensional deductive systems.
 Refinement by interpretation in  a general,  $k$-dimensional setting provides a  suitable context
to deal simultaneously with deductive systems arising from classes of models presented
in different logics, for example, as a set of equations, propositions or modal formulas.

\subsection{Scope.}\label{sc:scope}
Once stated the paper's contributions, it is important to clearly delimit its scope.
First of all it should be stressed that the focus of this paper is not placed on  \emph{speci\-fications}, understood as syntactic
entities which  describe in a structured, modular way classes of models, but rather  on the \emph{classes of models} themselves,
which constitute their denotations.  Deductive systems, the basic tool in our approach, correspond to such classes.
This means  that the whole area of \emph{specification structuring} \cite{ST06} is, for the moment,  left out.
Our approach is not concerned with the fact that specifications describing the relevant classes of models are  \emph{flat},
i.e. given by a finite set of sentences, or \emph{structured}, \emph{i.e.}, built by systematic application
of a number of operators, such as \emph{union},  \emph{translation} or \emph{hiding},
all of them well characterised in the literature and implemented in a number of computer-supported modelling tools.

This does not deny the fundamental importance of specification structuring.
 Research on this topic started with the introduction of \textsc{Clear}
 \cite{CLEAR79}, by the end of the seventies, and its role  cannot be underestimated.
Actually, the recursive definition of structured specifications provides
basic modular procedures for software composition and architecture.  Moreover structuring operations allows one to go
beyond the specification power of simple, unstructured specifications \cite{Bo02}.

Clearly, the approach proposed in this paper can be tuned to specification refinement in a strict sense.  In a recent publication
\cite{RMMB11} we showed how  \emph{refinement by interpretation} can be lifted to the level of structured specifications
with the usual operators mentioned above.
We believe, however, that by focussing on  \emph{classes of models}  this piece of research
acquires a broader scope of application and is worth on its own. In particular, it pays off when dealing with requirements that cannot
be properly formalised in a specification (for example, the property that a controller has a finite number of states).
Note this does not entail any  loss of expressivity:  for each specification, one may recursively compute its denotation
(a signature and a class of models) and work directly with them.

On the other hand, the discussion on which operators should be considered in a specification calculus is still
active. For example, very recently, reference \cite{RT11} introduced two new operators for specification
composition in order to deal with non-protecting importation modes. This further justifies the relevance of a semantic approach
as proposed here.

Another concept to make precise is \emph{refinement}. The word  is taken here in the broad sense of a
transformation mapping an abstract to a more concrete class of models. As such classes are represented by deductive systems,
a refinement will map a  deductive system into another, while preserving the consequence relation. This is in line with the usual
meaning of refinement: all requirements stated at the original level are still valid after refinement.
Moreover, it will be shown in the paper that  \emph{refinement by interpretation}
 between classes of models boils down to the standard notion of refinement as inclusion of classes of models
  whenever the witnessing interpretation is simply an identity.

Finally, a note on the expressivity of deductive systems. Actually, deductive systems can play a double role
representing both logics, on top of which all the specification
machinery can be developed, and classes of models as discussed above. Clearly, any institution induces, for each signature, a deductive system
through its satisfaction relation, and conversely, deductive systems may be viewed as special cases of institutions
as discussed in \cite{Vou2003}. For example, a deductive system may
represent the class of Boolean algebras; but the latter can also be specified as a theory in a suitable institution.
Another  example is provided by modal logics which can be regarded as both a deductive system  or
a theory in the first-order (FOL) institution through the standard translation.

This provides a uniform view of seemingly different settings, enabling us to
 discuss how essentially the same conceptual tool, that of \emph{logical interpretation}, can be used to
 interrelate logics and refine classes of  models both regarded as deductive systems.
 Reciprocally, deductive system can be endowed with an algebraic semantics, as
 discussed in section \ref{sc:algsem}. Note that the quest for such a uniform representation of logics and
 logical theories pops up in other contexts, namely on the design of logical frameworks. A prime example
 is provided by the logics-as-theories approach proposed by F. Rabe in \cite{RabePhD08}, resorting to a
 type theoretical framework, and further developed in the context of the \textsc{Latin} project \cite{latin11}.

 \subsection{Paper structure.}
Section \ref{sc:pre} introduces \emph{$k$-dimensional deductive systems} and their semantics following \cite{pigozzi}.
This paves the way to the formulation of refinement by interpretation in  a general setting in sections \ref{sc:gen}
 and \ref{sc:rbi2}. Before that, however, in section  \ref{sc:rbi1}, the approach is instantiated for the  case
 of algebraic specifications over the institution of Horn clause logic. This is a
 popular framework for algebraic specifications which not only deserves attention on its own, but also provides a
 simpler setting to build up intuitions.
  Finally, section \ref{sc:conc} concludes and suggests some problems
 deserving further attention.

\section{Preliminaries}\label{sc:pre}
\begin{center}
\begin{minipage}{0.8\textwidth}
\emph{Specifications of complex systems resort to different logics, and even to their
combination. Consequently a cha\-ra\-cterisation of
\emph{refinement by interpretation} needs to be orthogonal to whatever logic is used in specifications.
This is achieved through the notion of $k$-\emph{dimensional deductive systems}, of which the equational case is just an instance for $k=2$.
This section reviews such systems and their semantics, following \cite{pigozzi}, to provide the background for the sections to follow.}
\end{minipage}
\end{center}

%%%%%%%%%%%%%%%%%%%%%%
%\input{mmb12PRE.tex}

\subsection{Deductive systems and translations}
Roughly speaking, a deductive system is a general mathematical
  tool to reason about formulas in a language generated by a signature.
 Formally, it is defined as a   pair $\mc S=\langle \Sigma,\vdash \rangle$, where $\Sigma$ is a
 signature  and $\vdash$ is a substitution-invariant consequence relation between sets of formulas and individual formulas. Clearly, any standard sentential logical system, defined in the usual way by a set of axioms and a set of inference rules (for instance,
 classical and intuitionist propositional calculus, referred to in the sequel as CPC and IPC, respectively),
 is a deductive system.
   First order logic can also be formulated as a deductive system \cite{memoirs}, which shows
 how broad  the concept is.  The formal notion of a deductive system, in this abstract perspective,
  was originally considered by {\L}ukasiewicz  and  Tarski  \cite{Tarski} and intensely studied, from
  an  algebraic point of view, by many logicians.
 This gave rise  to a new, extremely relevant area of Mathematics, that of \emph{abstract algebraic logic} \cite{FJP_survey}.

Although in some  literature on algebraic logic this substitution-invariant consequence relation has been called a \emph{logic} (cf. \cite{proto}),
we adopt along the paper  the designation of \emph{deductive system} used by Blok and Pigozzi \cite{memoirs}.
 This terminology allows us to distinguish this concept from the habitual meaning  logic has in Computer Science,
typically understood  as an abstract  framework to express specifications and often
 abstracted as an \emph{institution} \cite{instituicoes,Dia08} or a $\pi$-institution \cite{FS88}.

As mentioned above, translation maps were introduced in the early 20th century as a means to interrelate
deductive  systems.
 They were first used  to understand the relationship between classical and constructive logics.   The well-known G\"odel translation of classical logic into intuitionistic logic has inspired disperse works on comparing different logics by means of translations. Illustrative examples include the works of Kolmogorov \cite{Kol77}, Glivenko \cite{G29}, and G{\"{o}}del \cite{go33} involving classical, intuitionist, and modal logics.

To the best of our knowledge the first general definition of translation between
deductive systems is due to Prawitz and Malmn\"as \cite{PM68}. More recently, W\'{o}jcicki \cite{Woj} presented a systematic study of translations between logics,  focussing on inter-relations between sentential logics.
And the quest goes on (cf. \cite{LT,CR,CCO09}).
At the turn of the century, Silva, D`Ottaviano and Sette \cite{SOF99}  proposed a general definition
 of translation between logics as maps preserving consequence
relations. Then, Feitosa and D`Ottaviano studied intensively the subclass of translations that preserve and reflect
consequence relations and coined the name \emph{conservative translation} \cite{traducoesconservativas}.

  Conservative translations  which are able to relate a formula to a set of formulas, and are therefore defined as
   \emph{multifunctions}, are  called \emph{interpretations}. Those which commute with substitutions
   were originally used in abstract algebraic logic to define a very important class of deductive systems ---
   referred to as  \emph{algebraisable} \cite{memoirs}.
  In particular, they  abstract the strict relationship between classic propositional logic and the class of Boolean algebras. A deductive system is said to be \emph{algebraisable} whenever there exists a class $K$ of algebras such that the consequence relation induced by $K$ is equivalent to
   the consequence in the deductive system. Such an equivalence was originally defined by means of two mutually inverse interpretations. Since then, this link between logic and universal algebra has been successfully explored. In particular, for an algebraisable \dlogic\ $\mc S$,
 properties of $\mc S$ can be related to algebraic properties of its equivalent algebraic semantics. This kind of results, of which many examples exist,  are often called \emph{bridge theorems}.

\subsection{$k$-dimensional deductive systems}
 In order to broaden the spectrum of application
 of deductive systems, Blok and Pigozzi et al. \cite{pigozzi}
introduced consequence relations over $k$-tuples of formulas, for $k$ a non-zero natural number.
   \emph{$k$-deductive systems},
   the result of this generalisation, are the higher dimensional version of the well known sentential logics.
   Their theory provides a unified treatment for several \dlogics\ such as the ones corresponding to
   assertional,  equational, and inequational logics. This generalisation also allows to regard interpretations
   witnessing algebraisability as a special kind of translations between $k$-deductive systems.

An \emph{equation},  represented in this paper by a formal
expression $\et$, can be regarded as a pair of terms (or formulas) $\langle t,t' \rangle$. This, in turn,
is an instance of a binary predicate standing for the equality of two terms. Similarly, a
\emph{unary predicate} asserting the validity of a formula is enough to represent a proposition.
The first observation leads to what will be characterised in the sequel as
a $2$-dimensional deductive system, of which the equational case is a particular instance.
The second corresponds, roughly speaking, to sentential logics  in a quite broad sense (to include,
for example, first-order predicate logic when suitably formalised).

In general, adding a $k$-ary predicate to a strict universal Horn theory without equality,
gives rise to a representation of a $k$-dimensional \dlogic, thus providing a suitable context
to deal simultaneously with different specification logics.

 We go even a step further considering $k$-deductive systems over \emph{many sorted} languages, because, in general,
 software systems manipulate several sorts of data. Almost all notions can be formulated in this broader setting
 as discussed later.
Note there are other generalisations that allow the reuse of
 arguments and tools from abstract algebraic logic in computer science contexts. Hidden logics,
 introduced by Pigozzi and Martins in \cite{conditional_prof} (see also \cite{TCS_mart} and \cite{Bm13}) are a
 prime example. They have been efficiently used  to develop specification and verification methodologies for  object oriented software systems. Examples include the Boolean logics, \ie, 1-dimensional multi-sorted logics with Bool
 as the only visible sort, and equality-test operations for some of the hidden sorts in place of equality predicates.

The syntactic support for  $k$-dimensional deductive systems is that of a  $k$-term.
Let $\Sigma = \pair{S,\Omega}$ be a signature and $X = (X_s)_{s \in S}$ a $S$-sorted set of variables.
A \emph{$k$-term}
of sort $s$ over  signature $\Sigma$
 is a sequence of $k$ $\Sigma$-terms, all of the same sort $s$,
 $\bar \varphi\:s=\langle \varphi_0\:s,\dots,\varphi_{k-1}\:s\rangle$, abbreviated to $\bar\varphi$  whenever
references to sorts can be omitted. A \emph{$k$-variable of sort $s$} is a sequence of $k$ variables all of the same sort $s$. $\kfm$ is the sorted set of all $k$-terms over $\Sigma$ with
variables in $X$, i.e.,
\[\Te^k_\Sigma(X) =  \langle (\Te_\Sigma(X)_{s})^k|s\in S \rangle\]

\noindent where $\Te_\Sigma(X)_{s}$ is the set of all terms over $\Sigma$, of sort $s$, with
variables in $X$.
Whenever each $\Te_\Sigma(X)_s$, for each sort $s$ in $\Sigma$, is non empty, their union acts as the carrier of  the $\Sigma$-\emph{term algebra} freely generated from $X$, which
we denote by $\Te_\Sigma(X)$. A \emph{substitution on} $\Te_\Sigma(X)$ is just an endomorphism over $\Te_\Sigma(X)$.

Let us fix some  notation and terminology: if
$\bar\varphi(x_0\:s_0,\dots,x_{n-1}\:s_{n-1})$ is a $k$-term over $\Sigma$, $A$ is a $\Sigma$-algebra,
 and $a_0\in
A_{s_0},\dots,a_{n-1}\in A_{s_{n-1}}$,  we denote by $\bar\varphi^A(a_0,\dots,a_{n-1})$
the value $\bar\varphi$ takes in $A$ when variables $x_0,\dots,x_{n-1}$ are instantiated
respectively by $a_0,\dots,a_{n-1}$. More precisely, if
$$\bar\varphi(x_0,\dots,x_{n-1})
=\langle\varphi_0(x_0,\dots,x_{n-1}),\dots,\varphi_{k-1}(x_0,\dots,x_{n-1})\rangle,$$ then
$\bar\varphi^A(a_0,\dots,a_{n-1})= h(\bar\varphi):=\langle
h(\varphi_0),\dots,h(\varphi_{k-1})\rangle$, where $h$ is any homomorphism from
$\Te_\Sigma(X)$ to $A$ such that $h(x_i)=a_i$ for all $i<n$.

Let $\Va=\langle \Va_s\rangle_{s\in S}$ be an arbitrary but fixed family of countably infinite disjoint sets
$\Va_s$ of variables of sort $s\in S$. Following a typical procedure in similar contexts (\eg, \cite{lew00}),
 we assume in the sequel $\Va$ fixed for
each set of sorts $S$ and large enough to contain all variables needed. Symbols of variables are obviously disjoint
of any other symbol in the signature. As usual in sentential logic frameworks, we will refer to
formulas ($k$-formulas) as  synonymous to terms ($k$-terms respectively). Accordingly, we will denote $\Te_\Sigma(\Va)$ by $\form{}{\Sigma}$.
Moreover, for each nonzero
natural number $k$, given a sorted signature $\Sigma$, a \emph{$k$-formula} of sort $s$ over
$\Sigma$ is any element of $(\Te_\Sigma^k(\Va))_s$. The set of all $k$-formulas will be
denoted by $\form{k}{\Sigma} $. Also note that an $S$-sorted subset $\Gamma$ of $k$-formulas
is identified with the unsorted set $\bigcup_{s\in S} \Gamma_s$, which allows writing  $\bar
\varphi\in \Gamma$ to mean that $\bar \varphi\in \Gamma_s$, for some sort $s$.
A set $\Gamma \subseteq \form{k}{\Sigma}$ is said to be \emph{globally finite} when $\Gamma_s$ is a finite set for each
sort $s$ of $\Sigma$, equal to  $\emptyset$ except for a finite number of them, \ie,  $\bigcup_{s\in S} \Gamma_s$ is finite.
In this setting, a $k$-dimensional deductive system
is defined as a
substitution-invariant consequence relation on the set of $k$-formulas.
The following definition generalises  the one due to W. Blok and D. Pigozzi  \cite{pigozzi} for the
one-sorted case.

\begin{defi}\label{df:klogic}
A \emph{$k$-dimensional deductive system}  is a pair $\cl=\left\langle \Sigma,\vdash_\cl
\right\rangle $, where $\Sigma$ is a sorted signature and $\vdash_\cl\,\subseteq\mathcal{P}(\form{k}{\Sigma})\times\form{k}{\Sigma} $ is a relation
such that, for all
$\Gamma\cup\Delta \cup\{\bar\gamma,\bar\varphi\}\subseteq \form{k}{\Sigma}$, the following
conditions hold:

\begin{enumerate}[label=(\roman*)]
\item $\Gamma\vdash_\cl\bar\gamma$ for each $\bar\gamma\in\Gamma$;
\item if $\Gamma\vdash_\cl\bar\varphi$, and $\Delta\vdash_\cl\bar\gamma$ for each $\bar\gamma\in\Gamma$,
then $\Delta\vdash_\cl\bar\varphi$;
\item if $\Gamma\vdash_\cl\bar\varphi$, then
$\sigma(\Gamma)\vdash_\cl \sigma(\bar\varphi)$ for every substitution $\sigma$.
\end{enumerate}

 A $k$-\dlogic\ is
\emph{specifiable} if $\vdash_\cl$ is \emph{compact} (or \emph{finitary} in the terminology
of abstract algebraic logic), \ie, if, whenever  $\Gamma\vdash_\cl\bar\varphi$, there exists
 a globally finite subset $\Delta$ of $\Gamma$
such that
$\Delta\vdash_\cl\bar\varphi$.
The relation $\vdash_\cl$, abbreviated to $\vdash$ whenever  $\cl$ is clear
from the context,  is called \emph{the consequence relation of} $\cl$.
\end{defi}

It is easy to see that, for any
$\Gamma\cup\Delta \cup\{\bar\gamma,\bar\varphi\}\subseteq \form{k}{\Sigma}$,
$\Gamma\vdash\bar\gamma$ and $\Gamma\subseteq \Delta$ imply
 $\Delta\vdash\bar\gamma$.

\medskip

Every consequence relation $\vdash$ has a natural extension to a relation between sets of  $k$-formulas, also denoted by
$\vdash$,  defined by $\Gamma\vdash\Delta$ if
$\Gamma\vdash\bar\varphi$ for each $\bar\varphi\in\Delta$. Finally, the relation of
\emph{interderivability}  between sorted sets is defined by $\Gamma\dashv\vdash
\Delta$ if $\Gamma\vdash \Delta$ and $\Delta\vdash \Gamma$. We  abbreviate
$\Gamma\cup\{\bar\varphi_0,\dots,\bar\varphi_{n-1}\}\vdash\bar\varphi$ and $\Gamma_0\cup
\dots\cup\Gamma_{n-1}\vdash\bar\varphi$ by
$\Gamma,\bar\varphi_0,\dots,\bar\varphi_{n-1}\vdash\bar\varphi$ and $\Gamma_0,
\dots,\Gamma_{n-1}\vdash\bar\varphi$, respectively.

Let $\cl$ be a (not necessarily specifiable) $k$-\dlogic. A \emph{thm} of $\cl$ is a
 $k$-formula $\bar\varphi$ such that $\vdash_\cl\bar\varphi$, \ie,
$\emptyset\vdash_\cl\bar\varphi$. The set of all theorems is denoted by $\Thm(\cl)$. An \emph{inference rule}
 is a pair $\langle
\Gamma,\bar\varphi\rangle$ where $\Gamma =\{\bar\varphi_{0},\dots,\bar\varphi_{n-1}\}$
a globally finite set of $k$-formulas and $\bar\varphi$
a $k$-formula, usually  represented as

\vspace{-0.05 cm}
\begin{equation}
\label{sequent} \displaystyle
\frac{\bar{\varphi}_{0},\dots,\bar{\varphi}%
_{n-1}}{\bar{\varphi}_{n}}
\end{equation}
\vspace{-0.05 cm}

A rule
such as (\ref{sequent}) is said to be a \emph{derivable rule} of $\cl$ if
$\{\bar\varphi_0,\dots,\bar\varphi_{n-1}\}\vdash_\cl \bar\varphi_n$. A set of $k$-formulas $T$
closed under the consequence relation, \ie, such that $T\vdash_\cl\bar\varphi$ implies $\bar\varphi \in
T$, is called a \emph{theory} of $\cl$. The set of all theories is denoted by $\Th(\cl)$; it
forms a complete lattice under set-theoretic inclusion, which is algebraic if $\cl$ is
specifiable. Given any set of $k$-formulas $\Gamma$, the set of all consequences of $\Gamma$,
in symbols $\Con_\cl(\Gamma)$, is the smallest theory that contains $\Gamma$. It is easy to
see that $\Con_\cl(\Gamma)=\{\,\bar\varphi\in\form{k}{\Sigma}: \Gamma\vdash_\cl
\bar\varphi\}$.
Often,  a specifiable $k$-\dlogic\ is presented in the so-called Hilbert style, \ie, by a
set of axioms ($k$-formulas) and inference rules.
 We say that a $k$-formula $\bar\psi$ is
\emph{directly derivable} from a set $\Gamma$ of $k$-formulas by a rule such as
(\ref{sequent}) if there is a substitution $h:\form{}{\Sigma} \rightarrow \form{}{\Sigma}$ such that
$h(\bar\varphi_n)=\bar\psi$ and $h(\bar\varphi_0),\dots,h(\bar\varphi_{n-1})\in\Gamma$.

Given a set $\AX$ of $k$-formulas and a set $\IR$ of inference rules, we say that $\bar\psi$
is \emph{derivable} from $\Gamma$ by $\AX$ and $\IR$, in symbols
$\Gamma\vdash_{\AX,\IR} \bar \psi$, if there is a \emph{proof}, \ie, a finite sequence of $k$-formulas,
$\bar\psi_0,\dots,\bar\psi_{n-1}$ such that $\bar \psi_{n-1}=\bar\psi$, and for each $i<n$, one
of the following conditions holds:

\begin{enumerate}[label=(\roman*)]
\item  $\bar\psi_i\in \Gamma$,
\item $\bar\psi_i$ is a substitution instance of a
$k$-formula in $\AX$,
 \item $\bar\psi_i$ is directly derivable from $\{\bar\psi_j:j<i\}$ by
one of the inference rules in $\IR$.
\end{enumerate}

It is clear that $\langle \Sigma,\vdash_{\AX,\IR}\rangle$ is a specifiable $k$-\dlogic.
Moreover, a $k$-deductive system $\cl$ is specifiable iff there exist possibly infinite sets $\AX$ and
$\IR$,  of axioms
 and inference rules respectively,  such that, for any $k$-formulas $\bar\psi$ and any set
$\Gamma$ of $k$-formulas, $\Gamma\vdash_\cl\bar\psi$ iff $\Gamma\vdash_{\AX,\IR} \bar
\psi$ (see \cite{proto} for the one sorted case). This justifies that all the examples of specifiable
deductive systems  introduced in this paper are presented by a  set of axioms and inference rules.

 If a deductive system $\cl$ is equal to $\langle \Sigma,\vdash_{\AX,\IR}\rangle$, for some sets $\AX$ and
$\IR$ with $|\AX\cup \IR|<\omega$, $\cl$ is said to be \emph{finitely axiomatisable}. A
$k$-\dlogic\   $\cl'=\langle \Sigma, \vdash_{\cl'} \rangle$ is an \emph{extension} of the
$k$-\dlogic\  $\cl=\langle \Sigma, \vdash_\cl \rangle$ if $\Gamma\vdash_{\cl'}\bar \varphi$
whenever $\Gamma\vdash_{\cl}\bar \varphi$ for all $\Gamma\cup \{\varphi\}\subseteq
\form{k}{\Sigma}$ (\ie, $\vdash_\cl\;\subseteq \; \vdash_{\cl'}$).
A $k$-\dlogic\  $\cl'$ is
\emph{an extension by axioms and rules of} a specifiable $k$-\dlogic\  $\cl$ if it can be
axiomatised by adding axioms and inference rules to the axioms and rules of some
axiomatisation of $\cl$.

 \subsection{The equational case}\label{sc:eqqq}
Typical examples of $k$-deductive systems are the ones induced by algebraic specifications (see Section \ref{sc:rbi1}). Given a signature
$\Sigma$, they are defined over pairs of $\Sigma$-terms $\langle t,t' \rangle$, representing equations  $t \approx t'$, and
have therefore dimension $k=2$.  Similarly, $\Sigma$-conditional equations can be taken as pairs
  $\langle \Gamma, e \rangle$ where $\Gamma$ is a globally
finite subset of $\form{2}{\Sigma}$ and $e\in\form{2}{\Sigma}$.
As a particular case, an  equation $\et$ is  a conditional equation without premisses, $\langle \emptyset, \et \rangle$.
In general, a conditional equation $\langle
\{ t_1\approx t'_1 ,\dots, t_n \approx t'_n\},\et \rangle$ will be written as $\ceq$.
In the sequel we will often use $\Eq(\Sigma)$, instead of  $\form{2}{\Sigma}$, for the set of all equations over $\Va$,
and, similarly, $\Ceq(\Sigma)$, for the set of  all $\Sigma$-conditional equations over $\Va$.

Let $\Gamma \cup \{t\approx t'\}\subseteq \form{2}{\Sigma}$ and $A$ be a $\Sigma$-algebra. We write
$\Gamma\models_At\approx t'$ if, for every homomorphism $h:\form{}{\Sigma}\rightarrow A$,
\begin{equation*}
h(\xi)=h(\eta), \text{ for every }\xi\approx\eta\in\Gamma, \text{ implies }h(t)=h(t').
\end{equation*}
For $\Gamma=\emptyset$, $\,\emptyset\models_A
t\approx t'$ is abbreviated to  $\models_At\approx t'$.
An equation $t\approx t'$ is an \emph{identity} \index{identity} of $A$ if $\models_A t\approx
t'$. Similarly, a  conditional equation
$\xi=t_1\approx t'_1 ,\dots, t_n \approx t'_n \rightarrow\et $
is a \emph{quasi}-\emph{identity} \index{quasi! identity} of $A$ if
$\{t_1\approx t'_1 ,\dots, t_n \approx t'_n\}\models_A \et$, which is simply written as $A \models \xi$ when
clear from the context.

These definitions extend, as expected,  to classes of algebras.
Given a class of $\Sigma$-algebras $K$, the (\emph{semantic}) \emph{equational}
\emph{consequence} \emph{relation} \index{equational! consequence relation} $\models_K$
\index{$\ecr$} determined by $K$ is  defined by
\begin{equation*}
\Gamma\models_K t\approx t'\text{ iff, for every } A\in K,\; \; \Gamma\models_A
 t\approx t'.
 \end{equation*}
In this case  $t\approx t'$ is  said to be a $K$-\emph{consequence} \index{consequence!
$\Kal$-} of $\Gamma$. When clear from the context,   we simply write
$K\models \xi$, where $\xi=t_1\approx t'_1 ,\dots, t_n \approx t'_n \rightarrow\et$,
 for $\{t_1\approx t'_1 ,\dots, t_n \approx t'_n\}\models_K t\approx t'$. Both $\models_A$ and  $\models_K$ extend to sets of
 equations $C$: $\Gamma \models_A C$ iff $\Gamma \models_A \xi$ for each $\xi \in C$, and
 respectively for $\models_K$. For a set $\Phi$ of quasi-equations, adopting the notational convention explained above and
 rather standard in Universal Algebra, $A \models \Phi$ stands for $A \models \xi$ for each $\xi \in \Phi$ (analogously for a class $K$ of $\Sigma$-algebras).

The equational consequence relation $\models_K$ satisfies the conditions
of Definition \ref{df:klogic}.
Hence it constitutes an example of a 2-\dlogic\  (perhaps the most important one!)
often simply denoted  by $K$.

 A class $K$ of $\Sigma$-algebras is axiomatised by a set $\Phi$ of conditional equations if

  \medskip

  \centerline{ $K=\Big\{A \; \mid\; \{t_1\approx t'_1 ,\dots, t_n \approx t'_n\} \models_A t\approx t' \; \; \mbox{ for all }
   t_1\approx t'_1 ,\dots, t_n \approx t'_n\rightarrow t\approx t'  \in \Phi\Big\}$.}

\medskip

\noindent It can be proved that, if $K$ is a class of $\Sigma$-algebras axiomatised
by a set $\Phi$ of conditional equations,
 then the relation $\models_{K}$ is specifiable (see \cite{blok} for the
one-sorted case). In this case it can be defined in the Hilbert style, taking
the set of $\Sigma$-equations  in $\Phi$, together with reflexivity, as the set of
axioms, and the set of $\Sigma$-conditional equations in $\Phi$, along with symmetry,
transitivity and congruence rules, as inference rules. Any specifiable equational
deductive system over $\Sigma$ is the natural extension (by axioms and rules) of the (2-dimensional) free
deductive system over  $\Sigma$ denoted by $\mathrm{EQ}_\Sigma$ and  defined in Figure \ref{fig:3}.
 Note that  the consequence relation associated to $\mathrm{EQ}_\Sigma $ coincides with  $ \models_{\mathrm \Alg(\Sigma)}$, where $\Alg(\Sigma)$
is the class of all $\Sigma$-algebras.

\medskip
\begin{figure}
\caixapeqg{\smallskip\begin{center}
\begin{minipage}{0.9\textwidth}
\begin{DSpecDefn}{$\mathrm{EQ}_\Sigma$}
\item[\Axioms]     ~
\\            $ \qquad \langle x\:s, x\:s\rangle$  \\
%\;\; \; \text{for each sort $s$}\\
\item[\textbf{inference rules}]  ~ \\
\\           $  \qquad\displaystyle \frac{\langle x\:s,y\:s\rangle}{\langle y\:s,
x\:s\rangle} \qquad$ \\
%for each sort $s$\\
\\           $   \qquad \displaystyle \frac{\langle x\:s, y\:s\rangle,\langle y\:s,
z\:s\rangle}{\langle x\:s, z\:s \rangle} \qquad$  \\
%for each sort $s$\\
\\            $ \qquad\displaystyle \frac{\langle x_0\:s_0, y_0\:s_0\rangle,\dots,
\langle x_{n-1}\:s_{n-1},
y_{n-1}\:s_{n-1}\rangle}{\langle O(x_0,\dots,x_{n-1}), O(y_0,\dots,y_{n-1})\rangle}$\\ ~\\ ~\\
\footnotesize{(for each sort $s$, operation symbol $O: s_0,\dots,s_{n-1}\rightarrow s$ in $\Sigma$)}
\end{DSpecDefn}
 \end{minipage}
\end{center}
\smallskip
}
\caption{ The 2-dimensional free deductive system over $\Sigma$.}\label{fig:3}
\end{figure}

Recall the discussion of the double role played by deductive systems in subsection \ref{sc:scope}.
Actually, it is well known that any deductive system can be seen as a $\pi$-institution \cite{FS88}, and
taking care of some
 extra technicalities, as an institution.
 However some special finitary ones are, in  a more natural way,  viewed as theories of institutions, \emph{i.e.}, as sets of sentences
 equipped with extra features ---  the deductive apparatus  they induce. For instance,  given a quasivariety of algebras, we can define a
 (2-dimensional) deductive system  over the set of equations by using the identities and quasidentities that define the quasivariety, together with the axioms of reflexivity and the inference rules of symmetry, transitivity and congruence rules.
 This deductive system can naturally be seen as a theory of the institution of Horn clause logic.

\subsection{$k$-structures}

As discussed in \cite{pigozzi}, a semantics for arbitrary $k$-\dlogics\  needs to go beyond
the usual algebraic structures, resorting to algebras endowed with a set of
$k$-tuples.
 Formally,

 \begin{defi}\label{df:kstruct}
 A \emph{$k$-structure} over a signature
$\Sigma = (S, \Omega)$ is a pair $\ca=\langle A, F\rangle$ where $A$ is a $\Sigma$-algebra and
$F$ is a sorted set $\pair{F_s}_{s \in S}$ such that $F_s \subseteq A_s^k$
%$F \subseteq A^k$.
\end{defi}

In this definition, the sorted set $F$ of designated elements of $A$, can be regarded
as a set of truth values on $A$: a formula holds if its interpretation is one of these elements.
This is why $F$ is called a \emph{filter}: for a deductive system representing the constructive
propositional calculus on a Boolean algebra the notion boils down to the familiar, Boolean filter.

Let $\mathcal{A}=\langle A,F\rangle$ be a $k$-structure. A $k$-formula $\bar \varphi\:v$
is said to be a \emph{semantic consequence} in $\mathcal{A}$ of a set of $k$-formulas
$\Gamma$, in symbols $\Gamma\models_{\mathcal{A}} \bar \varphi$, if, for every assignment
$h:\Va\rightarrow A$, $h(\bar \varphi)\in F_v$ whenever $h(\bar \psi)\in F_{w}$ for every $\bar
\psi\:w \in \Gamma$, where   $F_s$ is the component $s$ of the sorted set $F$.
  Notice that the same notation is used  for the assignment and its natural extension to formulas.
 A $k$-formula $\bar \varphi$ is \emph{valid} in $\mathcal{A}$, and
conversely $\mathcal{A}$ is a \emph{model} of $\bar\varphi$, if $\emptyset
\models_{\mathcal{A}}\bar\varphi$. A rule such as (\ref{sequent}) is a \emph{validity}, or a
\emph{valid rule}, of $\mathcal{A}$, and conversely $\mathcal{A}$ is a \emph{model} of the
rule, if $\{\bar\varphi_0,\dots,\bar\varphi_{n-1}\}\models_{\mathcal{A}} \bar\varphi_n$. A
formula $\bar\varphi$ is a \emph{semantic consequence} of a set of $k$-formulas $\Gamma$ for
an arbitrary class $\cM$ of $k$-structures over $\Sigma$, in symbols
$\Gamma\models_\cM\bar\varphi$, if $\Gamma\models_{\mathcal{A}}\bar\varphi$ for each
$\mathcal{A}\in \cM$. It can be proved that $\models_\cM$ is always a $k$-\dlogic, even if
 not always specifiable.

 Similarly, a $k$-formula or a rule is a \emph{validity of} $\cM$ if it is
a validity of each member of $\cM$. A $k$-structure $\mathcal{A}$ is a \emph{model} of a
$k$-\dlogic\  $\cl$ if $\Gamma\vdash_\cl\bar\varphi$ always implies $\Gamma\models_{\mathcal{A}}\bar\varphi$, \ie\
if every consequence of $\cl$ is a semantic consequence of $\mathcal{A}$.
The special models whose underlying algebra is the formula algebra, \emph{i.e.}, models of the form $\langle
\Fm^k(\Sigma), T\rangle$, for $T$ a theory, are called \emph{Lindenbaum-Tarski models}. The
class of all models of $\cl$ is denoted by $\Mod(\cl)$. If $\cl$ is a specifiable $k$-\dlogic,
then $\mathcal{A}$ is a model of $\cl$ iff every axiom and rule of inference is a validity of
$\mathcal{A}$.

In the equational case the semantics based on $2$-structures boils down to the traditional algebraic semantics. More precisely, given a quasi-equational class $K$ of $\Sigma$-algebras (\emph{i.e.} axiomatised by a set of quasi-equations over $\Sigma$), the algebra of any model of the 2-deductive system induced by $K$ with the identity as its filter belongs to $K$.

A class of $k$-structures $\cM$ is a \emph{$k$-structure semantics} of $\cl$ if
$\vdash_\cl\; = \;\models_\cM$. The class of all models of $\cl$ forms a $k$-structure
semantics of $\cl$. This fact is expressed in a specific completeness theorem \cite{conditional_prof},
stating that, for
any $k$-\dlogic\  $\cl$, $\Gamma\vdash_\cl\bar\varphi$ iff $\Gamma\models_{\Mod(\cl)}\bar\varphi$,
for every $\Gamma\cup\{\bar\varphi\}\subseteq \form{k}{\Sigma}$.

% ******************************
\section{Refinement by interpretation: The equational case}\label{sc:rbi1}

\begin{center}
\begin{minipage}{0.8\textwidth}
\emph{This section introduces, exemplifies and discusses the concept of  \emph{refinement by interpretation} ---
the core contribution of the paper --- framed in the specific, but popular, setting of algebraic specification over
the institution of Horn clause logic. }
\end{minipage}
\end{center}

%%%%%%%%%%%%%%%%%%%%%%%%%%%%%%%
%\input{mmb12RBI1.tex}

\subsection{Algebraic specification and refinement}
This section introduces refinement by logical interpretation for
 the specific case of algebraic specifications.
 As usual (see \eg, \cite{ST11}), an \emph{algebraic specification} $SP$  is considered here
 as a structured specification over the institution of Horn clause logic (\textbf{HCL})  \cite{Dia08}.
 Each  $SP$ denotes a pair $\langle \Sigma,\spm{SP}\rangle$ where $\Sigma$ is a signature and $\spm{SP}$ is a
class of $\Sigma$-algebras. The class $\spm{SP}$ of $\Sigma$-algebras is called the \emph{model class of}
$SP$, and each  $\Sigma$-algebra in $\spm{SP}$  a \emph{model of} $SP$. If $\xi$ is
a conditional equation $\langle \Gamma,e\rangle$ (respectively, an equation $e$), we write  $SP\models \xi$ for
$\Gamma\models_{\spm{SP}} e$ (respectively, $\models_{\spm{SP}} e$). Both cases extend, as expected, to sets of
conditional equations (respectively, equations).

When an algebraic specification $SP$ is
\emph{flat}, or \emph{basic}, its model class
 $\spm{SP}$ of algebras is axiomatised by a set $\Phi$ of
conditional equations. In this case $SP$ is the pair  $\langle \Sigma, \Phi \rangle$.
If the definition is restricted to formulas over a specific set of $\Sigma$-variables  $X \subseteq \Va$,
$SP$ is said $X$-\emph{flat}.
When $\Phi$ is a set of equations the flat specification $SP=\langle \Sigma, \Phi\rangle$ is called a (flat)
\emph{equational specification}. Recall that a class $K$ of $\Sigma$-algebras axiomatisable by a set of equations is called a \emph{variety}.
 A variety can also be characterised as a nonempty class $K$ of $\Sigma$-algebras
 closed under homomorphic images, subalgebras and direct products (cf. \cite{Burris}, \cite{ST11}). This famous result, due to
Birkhoff, turns out to be very useful to show that a given algebraic specification is not an equational specification.
 In the sequel, for simplicity, when clear from the context, we refer to algebraic specifications simply as specifications.

In this context  \emph{stepwise refinement} \cite{ST97, refinamentos_prof}  of specifications refers to the process through which
 a complex design is produced by incrementally adding details and reducing under-specification.
This proceeds step-by-step until the class of models becomes restricted to such an extent that a
  program can be easily manufactured.
  Formally, given a specification $SP$,  the implementation process
builds a correct realisation from a concrete enough  class of $\Sigma$-models $K$ such that $K$ is a subset of the class of denotations of
$SP$. During this
process,  the specification is enriched according to specific design decisions, iteratively approaching
the intended meaning for the final program.

Starting from an initial abstract specification $SP_0$, refinement builds a chain of specifications
$$SP_0\rightsquigarrow SP_1\rightsquigarrow SP_2\rightsquigarrow\cdots\rightsquigarrow SP_{n-1}\rightsquigarrow SP_n,$$
where, for $1\leq i\leq n$, $SP_{i-1} \rightsquigarrow SP_i$  represents
reverse inclusion of the respective classes of models.
Transitivity of inclusion assures that $SP_0\rightsquigarrow SP_n$.
From $SP_1$ onwards each element in this chain is the result of a
\emph{refinement step}.

In order to deal with more complex requirements along the implementation process, for example to
enable the possibility of renaming, adding or grouping together different
signature components, refinement steps $SP' \rightsquigarrow SP$ are traditionally taken up to
\emph{signature morphisms}. Recall that a \emph{signature morphism} from $\Sigma=\pair{S,\Omega}$ to
$\Sigma'=\pair{S',\Omega'}$
is a pair
$\sigma=\pair{\sigma_{sort},\sigma_{op}}$, where $\sigma_{sort}:S\rightarrow
S'$ and $\sigma_{op}$ is a $(S^*\times S)$-family of
functions respecting the sorts of operation names in $\Omega$, \emph{i.e.},
$\sigma_{op}=(\sigma_{\omega,s}:\Omega_{\omega,s}\rightarrow
\Omega'_{\sigma^*_{sort}(\omega),\sigma_{sort}(s)})_{\omega \in S^*,s \in S}$ (where for
$\omega=s_1\dots s_n \in S^*,
\sigma^*_{sort}(\omega)=\sigma_{sort}(s_1)\dots\sigma_{sort}(s_n)).$

Given a signature morphism $\sigma:\Sigma\rightarrow \Sigma'$ and a $\Sigma'$-algebra $A$,
let $A\!\upharpoonright_\sigma$ denote the \emph{reduct} of $A$ along $\sigma$, \ie, for any $s \in S$,
$s^{({A\!\upharpoonright_\sigma})} = \sigma(s)^{A}$, and for all $f:s_1,\dots,s_n \rightarrow s \in \Sigma$,
$f^{A\upharpoonright_\sigma}=\sigma_{op}(f)^{A}$. The notation $s^A$ and $f^A$ refer, respectively, to the carrier of sort $s$  and
the interpretation of operation symbol $f$ in the algebra $A$.
In this context we say that $SP'$ is a refinement of $SP$ witnessed by $\sigma$,
or simply a $\sigma$-refinement,
 when $ \mean{SP'} \!\upharpoonright_\sigma\,\subseteq\,  \mean{SP}$, where  $ \mean{SP'}\! \upharpoonright_\sigma = \{A\! \upharpoonright_\sigma|A \in \mean{SP'} \}$.

Having fixed the notation and the basic concepts we may now jump to the kernel of this section: representing classes
of specification models as 2-deductive systems and introduce interpretations as possible witnesses to the refinement steps.

\subsection{Denotations of algebraic specifications as 2-deductive systems}
As clarified in the Introduction (subsection \ref{sc:scope}), our approach is based on the
representation of specifications' model classes as deductive systems.
Our focus is essentially semantic, \emph{i.e.}, on
what specifications denote, and, therefore, the whole theory of refinement by interpretation discussed
here is independent from any specification structuring mechanism.

 Actually, any specification $SP$ denotes a  class of algebras --- its \emph{model class}, $\mean{SP}$.
 This in turn induces a 2-\dlogic\ according to the procedure sketched in subsection \ref{sc:eqqq}. Its
  consequence relation is $\models_{\spm{SP}}$.
The possibility of this consequence relation being non finitary, for example if  arising from non flat specifications,  is also covered in this approach.

As we will see, it turns out that $k$-deductive systems are an efficient universal tool to develop, in this way,
a theory about all classes of models over a fixed signature. A most important particular case is that of \emph{flat} specifications:
each of them can be identified
  with the deductive system it induces. This observation will be assumed throughout the text and in the examples.
  Moreover, using interpretations between
the induced deductive systems, one may go from one model class to another.

 In a sense to become clearer below, we say
 that a specification $SP'$ \emph{refines} another specification $SP$ \emph{by interpretation} $\tau$,
if $\tau$ is an interpretation of $SP$ such that
 $SP\models \xi \; \imp\; SP'\models \tau(\xi)$ for any formula $\xi$.
 Before jumping to the technical definition,
 consider the following example which, although elementary,
 may help to build up some intuitions for this notion of refinement.

\begin{exa}\label{firstex}

The introduction of a two-element Boolean algebra together with equality tests for each sort
allows the software engineer to formulate arbitrary universal first-order axioms as equations in the Boolean sort.
The importance of this construction comes from the fact that, although many of the most natural specification conditions
take the form of universal first-order sentences, only equational and conditional-equational axioms are guaranteed to possess
 an initial model. This motivation, as well as equality-test algebras in general, are extensively discussed in \cite{Pig91}.

Consider, thus, two flat specifications  $\mathrm{S}$  and $\mathrm{T}$. The former has a  signature $\Sigma$ which declares
 a sort $s$ and
a function $f: s \longrightarrow s$. $\mathrm{S}$ has an empty set of axioms; thus,  the corresponding class of models consists of all  algebras over its signature.
On the other hand, specification  $\mathrm{T}$ is depicted in Figure \ref{fig:unica}.

\begin{figure}[h]
\caixapeq{
\smallskip
\begin{center}
\begin{minipage}{0.86\textwidth}
\begin{SpecDefn}{T}
			\item[\Sorts]   ~
			\\              $s$
			\item[\Ops]     ~
			\\              $ok:\longrightarrow s$
			\\              $f:s \longrightarrow s$
			\\				$test:s \* s \longrightarrow s$
	\item[\textbf{axioms}]  ~  \\
	$test(x,x)\approx ok$ \\
         $ (test(x,x')\approx ok \,,\, test(x',x'') \approx ok ) \;  \rightarrow \; test(x,x'')\approx ok$\\
         $test(x,x') \approx ok \; \rightarrow\; test(x',x) \approx ok$\\
        $test(x,x') \approx ok \; \rightarrow \; test(f(x),f(x')) \approx ok$\\
\end{SpecDefn}
 \end{minipage}
\end{center}
\smallskip
}
\caption{Specification \texorpdfstring{$\mathrm{T}$}{T}.}\label{fig:unica}
\end{figure}

\noindent
\medskip
It is not difficult to see, by induction on the structure of proofs, that the translation
\begin{center}
\begin{tabular}{cccc}
$\tau:$ & $\Eq(\Sig(\mathrm{S}))$ & $\rightarrow$ & $\Eq(\Sig(\mathrm{T}))$  \\
 & $x\approx x' $ & $\mapsto $ & $test(x,x')\approx ok$  \\
\end{tabular}
\end{center}
\noindent
interprets $\mathrm{S}$ in the sense that the consequence relation induced by $\mathrm{S}$ is preserved and reflected by $\tau$.

An inspection of the signatures of both specifications shows that there exists an unique
signature morphism definable between them: the inclusion  $\iota:\Sig(\mathrm{S}) \rightarrow \Sig(\mathrm{T})$.
This morphism induces the identity translation between formulas which, obviously, does not interpret the specification above.

\end{exa}
\begin{flushright}
\vspace{-18 pt}$\Diamond$
\end{flushright}

\medskip

This kind of anomalies can be circumvented by tailoring refinement to the specific situation in hands.
In the example above equality may be taken as an extralogical predicate interpreted, roughly speaking, in the second specification as the test operation. However, such solutions just hold for special cases and have to be redefined case by case.
The approach proposed in this paper aims at providing a uniform, general solution.

\subsection{Translations}
A number of notions of translation between logical systems have been proposed in the
literature (see, for example, \cite{feitosatese,traducoesconservativas,pigozzi,LT}).

\begin{defi}[Translation] \label{df:eqtrans}
Let $\Sigma=(S,\Omega)$, $\Sigma'=(S',\Omega')$
 be two signatures.
 A \emph{translation from $\Sigma$ to $\Sigma'$}
 is a globally finite sorted multifunction\footnote{In the sequel,  notation $\mfdec{m}{A}{B}$ is used
 for multifunctions $m$ from $A$ to $B$.}
 $\mfdec{\tau}{\Eq(\Sigma)}{\Eq(\Sigma')}$. More precisely, $\tau$ maps each equation in $\Eq(\Sigma)$ into
 a globally finite $S'$-sorted set of equations in $\Eq(\Sigma')$.
  \end{defi}

When  $\Sigma=\Sigma'$,
$\tau$ is called a \emph{self translation of $\Sigma$}.
In this case, we say that $\tau$\emph{ commutes with
substitutions} if, for every substitution $\sigma$, and every equation $e\in \Eq(\Sigma)$,
$\tau(\sigma(e))=\sigma(\tau(e))$, where, again, the same notation is used to denote the application of a
substitution to a formula or a set of formulas.

Any translation $\mfdec{\tau}{\Eq(\Sigma)}{\Eq(\Sigma')}$ can be extended to conditional equations
as the  multifunction $\mfdec{\tau^*}{\Ceq(\Sigma)}{\Ceq(\Sigma')}$  given by

\begin{equation*}
\tau^*(\xi)=\{\langle \bigcup_{t\approx t'\in \Gamma} \tau (t\approx t'),e'\rangle \mid e'\in \tau(e)\}
\end{equation*}

\noindent
for $\xi=\langle \Gamma, e \rangle$ a conditional equation.
In the sequel, we  identify $\tau^*$ with $\tau$. The reason for
requiring that the image of $\tau_s$, for each sort $s$, to be globally finite becomes
now clear from  the definition of  $\tau^*$.
The following lemma establishes an important result.

\begin{lem}\label{propcze}
Let $\Sigma=(S,\Omega)$ be a standard signature
and
$\tau$ a self translation of $\Sigma$.
Then the following conditions are equivalent:
\begin{enumerate}[label=(\roman*)]
 \item  \label{1} $\tau$ commutes with substitutions.
 \item \label{2} There exist variables $x,y\in \Va$ and an $S$-sorted set of equations $E(x,y)$ in these two variables
   such that, for any $t\approx t' \in \Eq(\Sigma)_s$, $\tau_s(t\approx t')=E_s(t,t')$.
\end{enumerate}
\end{lem}
\proof
Assume that $\tau$ commutes with arbitrary substitutions, fix  distinct variables $x,y$, and
define $E:=\tau(x\approx y)$. Let $\FV(E)$ denote the set of free variables occurring in $E$,
and suppose $\FV(E)\subseteq\{x,y,r_1,r_2,\dots\}$. Let
$e$ be a substitution such that $e(x)=x$, $e(y)=y$, and $e(r_i)=x$ for all $i\geq 1$. By
assumption, $E(x,y,x,\dots)=E(e(x),e(y),e(r_1),...)=e(E)=e\big(\tau(x\approx y)\big)=\tau\big(e(x)\approx
e(y)\big)=\tau (x\approx y)=E(x,y,r_1,r_2,\dots)$. Hence, $\{r_1,r_2,\dots\}\subseteq\{x,y\}$.
Thus $\FV(E)\subseteq\{x,y\}$. We will write $E(x,y)$ to denote that the variables which occur in $E$ are among $x$ and $y$.
Now, let $t \approx t' \in \Eq(\Sigma)$, and $e$ be a
substitution such that $e(x)=t$ and $e(y)=t'$. We have that
$\tau(t \approx t')=\tau\big(e(x)\approx e(y)\big)=e\big(\tau(x\approx
y)\big)=e\big(E(x,y)\big)=E\big(e(x),e(y)\big)=E(t,t')$.
Suppose now that \eqref{2} holds. Let $\alpha$ be a substitution in $\Sigma$.  Then, for any $t
\approx t' \in  \Eq(\Sigma)$,
\begin{equation*}
\alpha(\tau(t \approx t'))\; =\;  \alpha(E_s(t,t'))\; =\; E_s(\alpha(t),\alpha(t'))) \; =\;
\tau(\alpha(t)\approx\alpha(t')) \; =\;  \tau(\alpha(t \approx t')).\eqno{\qEd}
\end{equation*}

\subsection{Interpretations}
Defined as a multifunction, a translation maps a formula into
 a set of formulas, and this is what makes translations interesting to establish relationships between
specifications. Recall that a signature morphism maps a term into just another term.

Not all translations, however,  are suitable to capture  the meaning of  interpreting a
specification into another. The following definition singles out the relevant ones:

\begin{defi}[Interpretation]Let $\tau$ be a translation from $\Sigma$ to $\Sigma'$,
and $SP$ a specification over $\Sigma$. The translation $\tau$ \emph{interprets $SP$} if there
is a class of algebras $K$ over $\Sigma'$  such that, for any $\xi \in \Ceq(\Sigma)$,
\begin{equation*}
SP\models \xi\;  \text{if and only if}\;  K\models \tau(\xi)
\end{equation*}
In this case we say that \emph{$\tau$
interprets $SP$ in $K$}, and \emph{$K$ is a $\tau$-interpretation of $SP$}.
Moreover, when $K$ is the denotation of a specification $SP'$ we say  that \emph{$\tau$
interprets $SP$ in $SP'$}, and \emph{$SP'$ is a $\tau$-interpretation of $SP$}.
\end{defi}

\begin{exa}\label{HvsB1}
The interpretation of the specification of
  \emph{Boolean algebras}  into the class HA of \emph{Heyting algebras} is a  classical example of an interpretation.
Let
$\Sigma$ be the usual signature of Boolean algebras
(and of Heyting algebras). Consider the well known \emph{double negation} (propositional)
map:
$\iota (t)\; =\; \neg\neg t$.

\noindent
Let $\tau$ be a self translation of $\Sigma$
defined by

\begin{equation*}
\tau(t\approx t')=\{\iota(t)\approx \iota(t')\}
\end{equation*}

\noindent
It can be shown that $\tau$ interprets the specification of  Boolean algebras in the class
 HA  \cite{blok}.
This is known as Glivenko's interpretation \cite{G29}. It
 establishes a strict relationship between classical and intuitionistic
derivability:  if $\varphi$ is a theorem in CPC then $\neg\neg \varphi$ is a theorem
in IPC. The result builds a bridge between the algebraic
semantics of \emph{classical propositional calculus} and that of
\emph{intuitionistic propositional calculus}.
\begin{flushright}
\vspace{-18 pt}$\Diamond$
\end{flushright}
\end{exa}

\noindent
It is not difficult to see that,

\begin{thm}
 Let $SP$ and $SP'$ be two algebraic specifications over a signature $\Sigma$, and
$\tau$ a recursive self translation of $\Sigma$
  that commutes with arbitrary substitutions and interprets $SP$ in $SP'$.
  If $SP$ is decidable then $SP'$ is decidable.
\end{thm}
\begin{proof}
In \cite{feitosatese}.
\end{proof}

\begin{defi}\label{df:tau:model} Let $\tau$ be a translation from $\Sigma$ to $\Sigma'$
and  $SP$ a specification over $\Sigma$. A $\Sigma'$-algebra $A'$ is a \emph{$\tau$-model of}
$SP$ if, for any  $\xi \in \Ceq(\Sigma)$, $SP\models \xi $ implies $A' \models \tau(\xi)$.
The \emph{$\tau$-model class of $SP$}, denoted by $\Mod_\tau(SP)$, is the class of all
$\tau$-models of $SP$.
\end{defi}
\noindent
Observe that, for any conditional equation $\xi=\langle \Gamma,e\rangle$, $SP\models\xi$  implies $\Mod_\tau(SP)\models \tau(\xi)$.

\CUT{
\begin{lem}
Let $SP=\langle \Sigma,\Phi \rangle$ be a $X$-flat specification and $\tau$ a \tra{X}{X'}. A
$\Sigma'$-algebra $A'$ is a $\tau$-model of $SP$ if $A' \models \tau(\Phi)$.
\end{lem}
}

\begin{thm}\label{axi0} Let $SP$ be a specification over $\Sigma$, and $\tau$
a translation from $\Sigma$ to $\Sigma'$ that interprets $SP$. Then the class of models
$\Mod_\tau(SP)$ is the largest $\tau$-interpretation of $SP$.
\end{thm}
\begin{proof}
Suppose that $\tau$ interprets $SP$. Let $K$ be a class of models which is a
$\tau$-interpretation of $SP$. Then for any $\xi \in \Ceq(\Sigma)$, $SP\models \xi$ if and
only if $K \models \tau(\xi)$. Hence all models in $K$ are $\tau$-models of $SP$. Thus,
$K \subseteq \Mod_\tau(SP)$.

So, we only need to prove that $\Mod_\tau(SP)$ is a $\tau$-interpretation of $SP$. Let $\xi \in
\Ceq(\Sigma)$. It is clear that  $SP\models \xi$ implies $\Mod_\tau(SP) \models \tau(\xi)$.
Suppose now that $\Mod_\tau(SP) \models \tau(\xi)$. Let  $K$ be a class of models that is a
$\tau$-interpretation of $SP$ (it exists since $\tau$ interprets $SP$). As proved above, $K
\subseteq \Mod_\tau(SP) $, hence $K\models \tau(\xi)$. Finally, $K$ being a $\tau$-interpretation
of $SP$ entails $SP \models \xi$.
\end{proof}

The next theorem states that if a specification $SP$ is finitely
axiomatisable, so is $\Mod_\tau(SP)$. Therefore there is a flat specification $SP^\tau$ such that $\spm{SP^\tau}= \Mod_\tau(SP)$.

\begin{thm}\label{axiomatizacao} Let   $\tau$ be a self translation of $\Sigma$
% w.r.t $X$
and  $SP=\langle \Sigma, \Phi \rangle$ a $X$-flat specification  for a set  $X \subseteq \Va$ of variables.
If $\tau$
commutes with substitutions then the class of models denoted by the flat specification $SP^\tau=\langle
\Sigma, \tau(\Phi)\rangle$ coincides with $\Mod_\tau(SP)$.
 Moreover, if $\Phi$ is finite then $SP^\tau$ is finitely axiomatisable.
\end{thm}
\begin{proof}
On the one hand, we have that, for any $A' \in \Mod_\tau(SP)$ and for any  $\xi\in
\Ceq(\Sigma)$, $SP\models \xi$ implies $A'\models \tau(\xi)$. In particular, since $SP
\models \Phi$, we have that $\Mod_\tau(SP) \models \tau(\Phi)$ and hence, $\Mod_\tau(SP)
\subseteq \mean{\langle \Sigma,\tau(\Phi) \rangle}$.

On the other hand, let $A \in \mean{\langle \Sigma,\tau(\Phi) \rangle}$ and $\xi=\langle
\Gamma,e\rangle$ be a conditional equation over $X$ such that $SP\models \xi$ (i.e., $\Gamma
\models_{\mean{SP}} e$). Then, it can be proved by induction on the length of a proof of
$e$ from $\Gamma$ in the deductive system $\models_{\mean{SP}}$ induced by $SP$, that $A\models \tau(\xi)$. Therefore $A$  is a
$\tau$-model of $SP$ and so, $\mean{\langle \Sigma, \tau(\Phi)\rangle} \subseteq
\Mod_\tau(SP)$.

Clearly, if $\Phi$ is finite then $\tau(\Phi)$ is also finite. Moreover, it can be proved that
$\tau(\Phi)$ constitutes an axiomatisation for $SP^\tau$.
\end{proof}

\subsection{Refinement by interpretation}\label{maindef}
Logical interpretation, as introduced in the previous section, provides the basic tool for the following definition:

\begin{defi}[Refinement by interpretation]
Let $SP$ be a specification over $\Sigma$ and $\tau$ a  translation from $\Sigma$ to $\Sigma'$
%\tra{X}{X'}
which interprets $SP$. We
say that a specification $SP'$ over $\Sigma'$ \emph{refines  $SP$ by
interpretation $\tau$}, in symbols $SP\rightharpoondown_\tau SP'$,
 if for any $\xi \in \Ceq(\Sigma)$
 \begin{equation*}
 SP\models \xi \; \imp\; SP'\models \tau(\xi)
 \end{equation*}
\end{defi}

It is not difficult to see that $SP'$ refines $SP$ by
interpretation $\tau$  whenever $\tau$ interprets $SP$ in $\mean{SP'}$.

Let us consider some examples of refinement by interpretation.  The first one is
mainly of theoretical interest: it shows how (a specification of) an Heyting algebra can
be regarded as a refinement of (a specification of) a Boolean algebra.

\begin{exa}
Consider the specifications of Boolean and Heyting algebras
depicted in Figures \ref{fig:4} and \ref{fig:5},
where $\mathrm{DISTLATTICE}$ is the specification of distributive lattices
(see, \cite{Burris}). We assume that a sort $bool$ is declared in $\mathrm{DISTLATTICE}$.
Note that in this example, as in a few others in the sequel,
 a naive use is made of a few standard operations  for structuring specifications.
 For example, annotation \textbf{enrich} $\mc S$ means that the (finitary)  specification
 is obtained  by adding new operation symbols, new sorts or/and new axioms to  $\mc S$.
 This is done just for syntactical convenience: to represent an equivalent flat specification, whose
 existence is trivially shown for both $\mathrm{BOOL}$,  $\mathrm{HEYTING}$ and all the other cases
 where we use the same artifice.  What should be kept in mind is that  the flat specifications
 correspond directly to deductive systems.

\begin{figure}
\caixapeq{\smallskip\begin{center}
\begin{minipage}{0.8\textwidth}
\begin{SpecDefn}{BOOL}
\item[\Enrich]
DISTLATTICE
%\item[\Sorts]   ~
%\\              $bool$
\item[\Ops]     ~
\\              $\fdec{\true}{~}{bool}$
\\              $\fdec{\false}{~}{bool}$
\\              $\fdec{\fuc{\neg}}{bool}{bool}$
\\              $\fdec{\e}{bool \times bool}{bool}$
\\              $\fdec{\ou}{bool \times bool}{bool}$
\item[\Axioms]  ~
\\              $p \ou \neg p \approx \true$
\\              $p \e \neg p \approx \false$
\\              $p \ou \true  \approx \true$
\\              $p \e \false  \approx \false$
\end{SpecDefn}
 \end{minipage}
\end{center} \smallskip
}
\caption{A specification of Boolean algebras.}\label{fig:4}
\end{figure}

\begin{figure}
\caixapeq{\smallskip\begin{center}
\begin{minipage}{0.8\textwidth}
\begin{SpecDefn}{HEYTING}
\item[\Enrich]
DISTLATTICE
%\item[\Sorts]   ~
%\\              $bool$
\item[\Ops]     ~
\\              $\fdec{\true}{~}{bool}$
\\              $\fdec{\false}{~}{bool}$
\\              $\fdec{\fuc{\neg}}{bool}{bool}$
\\              $\fdec{\e}{bool \times bool}{bool}$
\\              $\fdec{\ou}{bool \times bool}{bool}$
\\              $\fdec{\twoheadrightarrow}{bool \times bool}{bool}$
\item[\Axioms]  ~
\\              $p \ou \true \approx \true$
\\              $p \e \false \approx \false$
\\              $p \twoheadrightarrow p  \approx \true$
\\              $(p \twoheadrightarrow q) \e q  \approx q$
\\              $p \twoheadrightarrow (q \e r)  \approx (p \twoheadrightarrow q) \e (p \twoheadrightarrow r)$
\\              $p \e (p \twoheadrightarrow q)  \approx p \e q$
\\              $(p \ou q) \twoheadrightarrow r  \approx (p \twoheadrightarrow r) \e (q \twoheadrightarrow r)$
\\              $\neg p \approx p \twoheadrightarrow \false$
\end{SpecDefn}
 \end{minipage}
\end{center}
\smallskip
}
\caption{A specification of Heyting algebras.}\label{fig:5}
\end{figure}

\noindent
As in Example \ref{HvsB1},  the multifunction $\tau$ defined by
\begin{equation*}
\tau(p\approx q)\; =\; \{\neg\neg p \approx \neg\neg q\}
\end{equation*}
interprets BOOL in HEYTING. To show that
BOOL $\rightharpoondown_\tau$ HEYTING just observe that for any axiom
$\varphi$ of BOOL, HEYTING $\models \tau(\varphi)$.

A G{\"{o}}del algebra, also known as a $L$-algebra \cite{BD74},
 is a Heyting algebra that satisfies the pre-linearity condition: $(x \twoheadrightarrow y) \vee (y \twoheadrightarrow x) = \true$.
The class $\mathrm{GODEL}$  of all  G{\"{o}}del algebras  forms a subvariety of the class $\mathrm{HEYTING}$ and is thus a denotation of a flat specification, also denoted by $\mathrm{GODEL}$,  obtained from the $\mathrm{HEYTING}$  specification by adding the extra axiom $(x \twoheadrightarrow  y) \vee (y \twoheadrightarrow x) = \true $. So $\mathrm{GODEL}$  is an example of a refinement by interpretation of the specification of Boolean algebras.
\begin{flushright}
\vspace{-18 pt}$\Diamond$
\end{flushright}
\end{exa}

 Our next example, although quite elementary, illustrates a key point. It shows
 how  refinement by interpretation may
 capture \emph{data encapsulation}, \ie, the process of hiding a specific sort in a
 specification. This is a relevant issue in algebraic specification, in particular when the implementation
 target is an object-oriented framework: hidden sorts become the state space of object implementations,
 as discussed in, \eg,  \cite{Fav98,DW05}.
 In the following example a specification of the natural
numbers is interpreted into another one exclusively axiomatised  by equations of sort $bool$.
Sort $nat$ becomes hidden, or encapsulated, after refinement.

\begin{exa}
\label{natexamplemplo}
Consider the specification NAT of the natural numbers depicted in  Figure \ref{fig:6}. An alternative specification, NATEQ, is
shown in  Figure \ref{fig:7}, which  introduces an equality
 test, $\fuc{eq}$, axiomatised with the congruence property.

\begin{figure}
\caixapeq{\smallskip
\begin{center}
\begin{minipage}{0.8\textwidth}
\begin{SpecDefn}{NAT}
%\item[\Enrich] ~  $\EQ_{\Sigma_{nat}}$
\item[\Sorts]   ~
\\              $nat$
\item[\Ops]     ~
\\              $\fdec{\fuc{z}}{~}{nat}$
\\              $\fdec{\fuc{s}}{nat}{nat}$
\\              $\fdec{+}{nat \times nat}{nat}$
\item[\Axioms]  ~\\
$x + \fuc{z}\;  \approx \; x$ \\
$\fuc{s}(x+y) \; \approx\; x +  \fuc{s}(y) $\\
$\fuc{s}(x) \approx \fuc{s}(y) \; \rightarrow \; x \approx y $\\
%\\               $\displaystyle \frac{\fuc{s}(x) \approx \fuc{s}(y)}{x \approx y}$
\end{SpecDefn}
 \end{minipage}
\end{center}
\smallskip
}
\caption{A specification of the natural numbers.}\label{fig:6}
\end{figure}

\begin{figure}
\caixapeq{
\smallskip
\begin{center}
\begin{minipage}{0.8\textwidth}
\begin{SpecDefn}{NATEQ}
\item[\Enrich]  BOOL
\item[\Sorts]   ~
\\              $nat$
\item[\Ops]     ~
\\              $\fdec{\fuc{z}}{~}{nat}$
\\              $\fdec{\fuc{s}}{nat}{nat}$
\\              $\fdec{+}{nat \times nat}{nat}$
\\              $\fdec{\fuc{eq}}{nat \times nat}{bool}$
\item[\Axioms]  ~
\\              $\fuc{eq}(x,x) \approx \true$
\\              $\fuc{eq}(x,y) \approx \true \; \impp\; \fuc{eq}(y,x) \approx \true$
\\              $\fuc{eq}(x,y) \approx \true \; ,\; \fuc{eq}(y,\fuc{z}) \approx \true \;  \impp\; \fuc{eq}(x,\fuc{z}) \approx \true$
\\              $\fuc{eq}(x,y) \approx \true \; \impp\; \fuc{eq}(\fuc{s}(x),\fuc{s}(y)) \approx \true$
\\              $\fuc{eq}(\fuc{s}(x),\fuc{s}(y)) \approx \true \; \impp\; \fuc{eq}(x,y) \approx \true$
\\              $\fuc{eq}(x + \fuc{z}, x) \approx \true$
\\              $\fuc{eq}(\fuc{s}(x+y), x +  \fuc{s}(y) ) \approx \true$\\
%\item[\textbf{inference rules}]  ~\\
%\\               $\displaystyle \frac{\fuc{eq}(x,y) \approx \true}{\fuc{eq}(y,x) \approx \true}$
%\hspace{2mm}
%   $\displaystyle \frac{\fuc{eq}(x,y) \approx \true}{ \fuc{eq}(\fuc{s}(x),\fuc{s}(y)) \approx \true}$  ~ \\ \vspace{3mm}
%\\  $\displaystyle \frac{\fuc{eq}(x,y) \approx \true,\; \fuc{eq}(y,z) \approx \true}{\fuc{eq}(x,z) \approx \true}$
%\hspace{2mm}
% $\displaystyle \frac{\fuc{eq}(\fuc{s}(x),\fuc{s}(y)) \approx \true}{\fuc{eq}(x,y) \approx \true}$
\end{SpecDefn}
 \end{minipage}
\end{center}
\smallskip
}
\caption{Hiding sort $nat$.}\label{fig:7}
\end{figure}

\noindent
NATEQ interprets NAT through multifunction $\tau$ defined as
\begin{equation*}
\tau(x:nat \approx y:nat)=\{\fuc{eq}(x:nat,y:nat) \approx \true \}
\end{equation*}
  First note that for any equation $t \approx t'$
such that NAT $\models t \approx t'$, one also has NATEQ $\models
\fuc{eq}(t,t') \approx \true$, since the interpretation of the proof of NAT $\models
t \approx t'$ is a proof of
NATEQ $\models \fuc{eq}(t,t') \approx \true$. The converse is proved  by induction
on the length of the proof of NATEQ $\models \fuc{eq}(t,t') \approx \true$. Hence, any refinement of NATEQ,
for example, the one obtained by adding axiom $\fuc{eq}(\fuc{z}, \fuc{s}(\fuc{z})) \approx \false$, is
a refinement by $\tau$ of NAT.
\begin{flushright}
\vspace{-18 pt}$\Diamond$
\end{flushright}
\end{exa}

Other useful design transformations  can similarly be captured as
 refinements. Our last example illustrates one of them in which
 some operations are \emph{decomposed} or \emph{mapped to transactions}, \ie,
 sequences of operations to be executed atomically.

\begin{exa}\label{ex:bams1}

Consider the following fragment of a specification of a bank account management system (BAMS),
involving account deposits (operation $\dep$), withdrawals ($\withd$), and a balance query ($\balance$).

\begin{center}
\begin{minipage}{0.8\textwidth}
\begin{SpecDefn}{BAMS}
\item[\Enrich]  INT
%=  \emph{enrich} INT \emph{with}
\item[\Axioms]  ~
\\              $\balance(\dep(s,i,n),i) \; \approx\; \balance(s,i) + n $
\\              $ \balance(\withd(s,i,n),i) \; \approx \; \max(\balance(s,i) - n, 0)$
\\              $\cdots$
%\item[\textbf{inference rules}]  ~
%\\              $\frac{ n\:Int \approx n'\:Int}{ s(n)\approx s(n')}$
%\\              $\cdots$\\
\end{SpecDefn}
%\begin{SpecDefn}{BAMS}= BAMS0 + $ EQ_{\Sigma_{BAMS0}}$
%\end{SpecDefn}
 \end{minipage}
\end{center}

Assume INT as the usual flat specification of integer numbers with
arithmetic operations, and variables $s\:Sys$, $i\:Ac$ and $n,n'\:Int$, where $Sys$ and $Ac$ are the sorts of bank systems and
account identifiers, respectively. The signatures of the main operations are as follows:
 $\dep : Sys \times Ac \times Int\; \longrightarrow\; Sys$,
 $\withd :  Sys \times Ac \times Int\; \longrightarrow\; Sys$ and
 $\balance :  Sys \times Ac \; \longrightarrow\; Int$.

\noindent
Consider now an implementation
B2 where all debit and credit transactions require a previous validation
step. This is achieved through an operation $\transvalid : Sys \times Ac \times Int\; \longrightarrow\; Int$, which, given
a bank system state, an account identifier, and a value to be added or subtracted to the account's balance, verifies if
the operation can proceed or not. In the first case it outputs the original amount; in the second $0$ is returned as an error value.
This will force an invalid deposit or withdrawal to have no effect ($0$ will be added to, or subtracteded from the account's balance).
Although this is a quite common form of error recovery (invoking the intended operation with its identity element),
note that, for the purpose of this example, the concrete behaviour of operation $\transvalid$ is not relevant as long as an integer
is returned.
The axioms for B2 include,

\begin{center}
\begin{minipage}{0.8\textwidth}
\begin{SpecDefn}{B2}
\item[\Enrich]  INT
\item[\Axioms]  ~
\\              $\cdots$
\\              $ \balance(\dep(s,i,\transvalid(s,i,n)),i) \approx  \balance(s,i) + \transvalid(s,i,n) $
\\              $ \balance(\withd(s,i,\transvalid(s,i,n)),i) \; \approx \;
\max(\balance(s,i) - \transvalid(s,i,n), 0) $

\end{SpecDefn}
 \end{minipage}
\end{center}
~\\

\noindent
The interpretation
$\mfdec{\tau_1}{\eqq{\mathrm{BAMS}}}{\eqq{\mathrm{B2}}}$, defined by

\begin{equation*}
	\tau_1(t\approx t')=\big\{\gamma \approx \gamma' | \gamma\in \tau_1^\#(t) \mbox{ and } \gamma' \in \tau_1^\#(t')\big\}, \text{ where }
%	\tau^\#_1(x)=\{x\},\; \; \text{for}\; x\in \Va
\end{equation*}

\begin{align*}
&\tau^\#_1(x)\; =\; \{x\}\; \; & \text{for}\; x\in \Va\\
&\tau^\#_1 (f(t_1,t_2,t_3)) \;  =\;\big\{f(t'_1,t'_2,\transvalid(t'_1,t'_2,t'_3))\; \mid \bigwedge_{i=1..3}\, t'_i \in \tau_1^\#(t_i)\big\}& \text{for}\; f\in \{\dep, \withd\}\\
&\tau^\#_1 (f(t_1,\dots,t_n))\; =\; \big\{f(t'_1,\dots,t'_n)\; \mid \bigwedge_{i=1..n}\, t'_i \in \tau_1^\#(t_i)\big\}\;  & \text{for}\; f\notin \{\dep, \withd\}
\end{align*}

\noindent
witnesses a refinement in which isolated calls to
the operations  are mapped to validated transactions.

It might also be the case that some operations  can be executed
with or without validation, as in, for example,

\begin{center}
\begin{minipage}{0.8\textwidth}
\begin{SpecDefn}{B3}
\item[\Enrich]  INT
\item[\Axioms]  ~
\\              $\cdots$
\\              $\balance(\dep(s,i,\transvalid(s,i,n)),i) \, \approx\, \balance(s,i) +  \transvalid(s,i,n)$
\\              $ \balance(\dep(s,i,n),i) \, \approx\, \balance(s,i) + n$
\\              $ \balance(\withd(s,i,\transvalid(s,i,n)),i) \, \approx\, \max(\balance(s,i) - \transvalid(s,i,n), 0) $
\end{SpecDefn}
 \end{minipage}
\end{center}
~\\

\noindent
Clearly, $\mathrm{B3}$ refines
$\mathrm{BAMS}$ through the  interpretation $\mfdec{\tau_2}{\eqq{\mathrm{BAMS}}}{\eqq{\mathrm{B3}}}$,

\begin{align*}
&\tau^\#_2(x)\; =\; \{x\}\; \; \; \text{for}\; x\in \Va\\
&\tau^\#_2 (\dep(t_1,t_2,t_3)) \;  =\;\big\{\dep(t'_1,t'_2,\transvalid(t'_1,t'_2,t'_3)), \dep(t'_1,t'_2,t'_3)\; \mid \bigwedge_{i=1..3}\, t'_i \in \tau_2^\#(t_i)\big\} \\
&\tau^\#_2 (\withd(t_1,t_2,t_3)) \;  =\;\big\{\withd(t'_1,t'_2,\transvalid(t'_1,t'_2,t'_3))\; \mid \bigwedge_{i=1..3}\, t'_i \in \tau_2^\#(t_i)\big\}\\
&\tau^\#_2 (f(t_1,\dots,t_n))\; =\; \big\{f(t'_1,\dots,t'_n)\; \mid \bigwedge_{i=1..n}\, t'_i \in \tau_2^\#(t_i)\big\}\;  \; \text{for}\; f\notin \{\dep, \withd\}
\end{align*}

\noindent
Finally, the reader is invited to check that
 B4 is  a refinement of BAMS through interpretation
$\tau_3:\Eq(\Sigma_{\mathrm{BAMS}})\rightarrow \Eq(\Sigma_{\mathrm{B4}})$, defined similarly to $\tau^\#_2$ but for the $\withd$ case which
becomes

\begin{align*}
\tau^\#_3 (\withd(t_1,t_2,t_3)) \;  =\;\big\{\perm(\withd(t'_1,t'_2,\transvalid(t'_1,t'_2,t'_3)))\; \mid \bigwedge_{i=1..3}\, t'_i \in \tau_3^\#(t_i)\big\}\\
\end{align*}

\noindent
where the specification $\mathrm{B4}$ not only forces all operations to be validated, but
may also require an additional step to check the integrity of the account state before a debit operation is executed.
This step is abstracted in an operation $\perm : Sys \; \longrightarrow\; Sys$. Therefore, operation $\withd$
is  decomposed into a three  step transa\-ction:
\begin{center}
\begin{minipage}{0.8\textwidth}
\begin{SpecDefn}{B4}
\item[\Enrich]  INT
\item[\Axioms]  ~
\\              $\cdots$
\\              $ \balance(\dep(s,i,\transvalid(s,i,n)),i) \; \approx\; \balance(s,i) + \transvalid(s,i,n)$
\\              $\balance(\dep(s,i,n),i) \;\approx\; \balance(s,i) + n$
\\              $\balance(\perm(\withd(s,i,\transvalid(s,i,n))),i)  \approx  \max(\balance(s,i) - \transvalid(s,i,n), 0) $
\end{SpecDefn}
 \end{minipage}
\end{center}

\begin{flushright}
\vspace{-18 pt}$\Diamond$
\end{flushright}
\end{exa}

%*********************************************
%*********************************************

%*********************************************

\subsection{Stepwise refinement  revisited}\label{sc:mainth}

Having illustrated some typical applications of  refinement by interpretation,
it is legitimate to ask  how does it relate to classical refinement
based on signature morphisms. Let us first recall the standard definition:

\begin{defi}[$\sigma$-Refinement]
\index{Refinamento!$\sigma$-Refinamento} Let $\sigma:\Sigma \rightarrow \Sigma'$ be a
signature morphism. The specification $SP'$ over $\Sigma'$ is a\emph{ $\sigma$-refinement of}
$SP$, in symbols $SP\rightsquigarrow_\sigma SP'$, if

$$\mean{SP'} \!\upharpoonright_\sigma\,\subseteq\, \mean{SP} $$

\noindent where $\mean{SP'}\!\upharpoonright_\sigma = \{A\!\upharpoonright_\sigma|A \in
\mean{SP'}\}$.   We write  $SP\rightsquigarrow SP'$ whenever it is witnessed by the identity.
\end{defi}

Note that, as we associate a fixed set of variables $\Vav$ to each signature, the usual notion of signature morphism
 has to be extended so that terms with variables from $\Vav$ can still be
handled. Therefore we assume that each signature morphism $\sigma:\Sigma \rightarrow \Sigma'$
has a component $\sigma_{var}$ on variables, namely an injective mapping from $\bigcup_{s\in S} \Vav_s$
to $\bigcup_{s\in S'} \Vav_s'$ such that, for every variable $x\in \Vav_s$,  $\sigma_{var}(x)\in \Vav_{\sigma_{sort}(s)}$.
This allows for the definition of $\sigma^*$, the extension of $\sigma$ to terms given by $\sigma^*(x) = \sigma_{var}(x)$,
for each variable $x$, and by $\sigma^*(f(t_1,\dots,t_n))=\sigma_{op}(f)(\sigma^*(t_1),\dots,\sigma^*(t_n))$,
  for each term $f(t_1,\dots,t_n)$ in $\Sigma$.

Since the
composition of two signature morphisms is still a signature morphism,
refinements can be composed vertically: if
$SP_0\rightsquigarrow_{\sigma_1} SP_1$ and $SP_1\rightsquigarrow_{\sigma_2} SP_2$ then
$SP_0\rightsquigarrow_{\sigma_2\circ\sigma_1} SP_2$.
This is simply a consequence of reducts being functorial.

The so called  \emph{satisfaction lemma} \cite{instituicoes} is at the basis of classical refinement.
Actually, theorem \ref{refmotiva} below, whose proof relies on the satisfaction lemma,
provides an important characterisation of $\sigma$-refinements.
\begin{lem}[Satisfaction Lemma]
\label{satisfacao}  Let $\Sigma$ and  $\Sigma'$ be signatures,
$\sigma:\Sigma\rightarrow\Sigma'$ a signature morphism, $A'$ a $\Sigma'$-algebra, and $\xi$
a conditional equation. Then,
\begin{equation*}
A'\models \sigma(\xi)\; \text{\mbox{ iff }}\;  A'\!\upharpoonright_\sigma \models \xi
\end{equation*}
\end{lem}

\begin{thm}
\label{refmotiva} Let $\sigma:\Sigma\rightarrow \Sigma'$ be a signature morphism, $SP=\langle
\Sigma, \Phi\rangle $ a $X$-flat specification and $SP'$ a  specification over $\Sigma'$.
Then, $SP\rightsquigarrow_\sigma SP'$ iff $SP'\models \sigma(\Phi)$.
\end{thm}

\begin{proof}
Suppose that $SP\rightsquigarrow_\sigma SP'$. Then, for any $A' \in
\mean{SP'}$, $A'\!\upharpoonright_\sigma \in
\mean{SP}$, i.e., $A'\!\upharpoonright_\sigma
\models \Phi$. Hence, by Lemma \ref{satisfacao}, $A'\models
\sigma(\Phi)$. On the other hand, suppose  $SP'\models
\sigma(\Phi)$. Then, for any $A' \in \mean{SP'}$, $A' \models \sigma(\Phi)$. By Lemma
\ref{satisfacao} $A'\!\upharpoonright_\sigma \models \Phi$, and hence,
$A'\!\upharpoonright_\sigma \in \mean{SP}$. Therefore
$SP\rightsquigarrow_\sigma SP'$.
\end{proof}

A relationship between $\sigma$-refinement and refinement by interpretation can now be described as follows.

\begin{thm}\label{carmot}
Let $SP$ be a specification over $\Sigma$, and $\tau$ a translation from $\Sigma$ to $\Sigma'$
% \tra{X}{X'}
which interprets $SP$. If there is a specification $SP^\tau$ whose  denotation coincides with $\Mod_\tau(SP)$, then,
 for every $SP'$  over $\Sigma'$,  $SP^\tau\rightsquigarrow SP'$ implies
$SP\rightharpoondown_\tau SP'$.
\end{thm}

\begin{proof}
Suppose $SP^\tau\rightsquigarrow SP'$, \emph{i.e.},  $\mean{SP'} \subseteq
\mean{SP^\tau}$. Thus, any algebra $A' \in \mean{SP'}$ is a $\tau$-model of $SP$.
Therefore for any $\xi \in \Ceq(\Sigma)$, $SP\models \xi$ implies $SP'\models \tau(\xi)$.
I.e., $SP\rightharpoondown_\tau SP'$.
\end{proof}

\begin{thm}
\label{carmot1} Let $SP $ be a specification over $\Sigma$ and $\tau$
a  translation from $\Sigma$ to $\Sigma'$.
If $\tau$ interprets $SP$ and  there is a specification $SP^\tau$ whose  class of models coincides with $\Mod_\tau(SP)$,
 then the following conditions are equivalent:
\begin{align}
  & SP\rightharpoondown_\tau SP' \label{one}\\
  & SP^\tau\rightsquigarrow SP' \label{two}
\end{align}
\end{thm}
\begin{proof}
Suppose that $SP\rightharpoondown_\tau SP'$. Since any $A \in \mean{SP'}$ is a $\tau$-model,
$\mean{SP'} \subseteq \mean{SP^\tau}$. Hence, $SP^\tau\rightsquigarrow SP'$.
The converse implication is just Theorem
\ref{carmot}.
\end{proof}

\noindent
An immediate corollary is

\begin{cor}
Let $SP=\langle \Sigma, \Phi \rangle$ be a $X$-flat specification and $\tau$ a translation from $\Sigma$ to $\Sigma'$
%\tra{X}{X'},
which interprets $SP$
and $\Mod_\tau(SP)$ is axiomatised by $\tau(\Phi)$.  Then,  $SP'\models \tau(\Phi)$ implies $SP\rightharpoondown_\tau SP'$.
\end{cor}
\begin{proof}
Let $SP^{\tau}$ be the flat specification $\pair{\Sigma', \tau(\Phi)}$. By hypothesis $\mean{SP^\tau} = \Mod_\tau(SP)$.
Suppose that $SP' \models \tau(\Phi)$. Then $\mean{SP'} \subseteq \mean{SP^\tau}$, which entails
$SP^\tau\rightsquigarrow SP'$. By Theorem \ref{carmot}, since $\mean{SP^\tau} = \Mod_\tau(SP)$, we conclude
$SP \rightharpoondown_\tau SP'$.
\end{proof}

This paves the way to address the  following question: \emph{given that
any mapping  can be regarded as a multifunction, when does a $\sigma$-refinement
via signature morphism become a refinement by interpretation?}

First note that a signature morphism  $\sigma:\Sigma\rightarrow \Sigma'$
induces a translation $\tau:\Ceq(\Sigma)\rightarrow \Ceq(\Sigma')$.
This is defined, for each   $\langle \Gamma, e \rangle \in \Ceq(\Sigma)$,  by
$\tau(\langle \Gamma, t\approx t' \rangle)=\langle \{\sigma(t)\approx \sigma(t')|t\approx t' \in \Gamma\},\sigma(t)\approx \sigma(t')\rangle$.
Then,

\begin{lem}
Let $SP$ be a specification over $\Sigma$, $\sigma:\Sigma \rightarrow \Sigma'$   \nt{an
injective} signature morphism, and $\tau$  the translation induced by the signature morphism
$\sigma$. Then $\tau$ interprets $SP$.
\end{lem}
\begin{proof}Let  $K = \{A'|A' \!\upharpoonright_\sigma \in \mean{SP}\}$. Let $\xi \in \Ceq(\Sigma)$, and
 suppose $SP\models \xi$. Let $A' \in K$. Since $A'\!\upharpoonright_\sigma
  \in  \mean{SP} $ we have that
$A'\!\upharpoonright_\sigma \models \xi$ and, by Lemma \ref{satisfacao},  $A'\models
\sigma(\xi)$. Therefore $K \models \tau(\xi)$. Suppose now $K \models \tau(\xi)$ and let $A
\in \mean{SP}$. Since $\sigma$ is injective there is $B \in K$ such
$B\!\upharpoonright_\sigma=A$. Thus  $B\models \tau(\xi)$ and, by Lemma \ref{satisfacao},
$B\!\upharpoonright_\sigma\models \xi$, \emph{i.e.}, $A\models \xi$. Hence $SP\models \xi$.
Therefore $\sigma$ interprets $SP$.
\end{proof}

Next theorem shows that refinement by interpretation is, actually, a generalisation of $\sigma$-refinement whenever the
 signature morphism $\sigma$ is injective.

\begin{thm}
Let $SP$ and $SP'$ be two specifications over $\Sigma$ and $\Sigma'$ respectively, and
$\sigma:\Sigma \rightarrow \Sigma'$  \nt{ an injective} signature morphism. Let $\tau$ be the
translation induced by the signature morphism $\sigma$. Then, $SP \rightsquigarrow_\sigma SP'$
implies $SP\rightharpoondown_\tau SP'$.
\end{thm}
\begin{proof}By the previous theorem $\tau$ interprets $SP$. Suppose $SP \rightsquigarrow_\sigma SP'$, \emph{i.e.},
$\mean{SP'} \!\upharpoonright_\sigma \subseteq \mean{SP}$.
 Let  $\xi \in \Ceq(\Sigma) $ such that
$SP\models \xi$. Let $A' \in \mean{SP'}$. Then
$A'\!\upharpoonright_\sigma \in \mean{SP}$ and so
$A'\!\upharpoonright_\sigma\models \xi$.
By Lemma \ref{satisfacao}, $A'\models \sigma(\xi)$. Hence $SP'\models \sigma(\xi)$. Therefore
$SP\rightharpoondown_\tau SP'$.
\end{proof}

\noindent
Finally, we show that, in the \emph{flat} case, the two concepts of refinement coincide:

\begin{thm}
Let $\sigma:\Sigma \rightarrow \Sigma'$ be  \nt{an injective} signature morphism, $SP=\langle
\Sigma, \Phi\rangle$ a $X$-flat specification and $SP'$ a specification over $\Sigma'$. Let
$\tau$ be the translation induced by the signature morphism $\sigma$ . Then
$SP\rightsquigarrow_\sigma SP'$ iff $SP\rightharpoondown_\tau SP'$.
\end{thm}
\begin{proof}

Suppose $SP\rightharpoondown_\tau SP'$.  Since $SP \models \Phi$, $SP' \models \sigma(\Phi)$.
By Theorem \ref{refmotiva} $SP\rightsquigarrow_\sigma SP'$.
\end{proof}

\noindent
It should be noted at this point that the discussion concerning composition of refinements by interpretation is not
straightforward. In fact, \emph{horizontal composition} is still an open question (see section \ref{sc:conc}).
For \emph{vertical composition} an additional property has to be
imposed on the components' interpretations. Formally,

\begin{thm}\label{th:ver1}
Let $SP$, $SP'$ and $SP''$ be three specifications over $\Sigma$, $\Sigma'$ and $\Sigma''$
respectively. Let $\tau$ be a translation from $\Sigma$ to $\Sigma'$
% \tra{X}{X'}
and $\rho$  a translation from $\Sigma'$ to $\Sigma''$
%\traa{X'}{X''}.
Suppose that
$SP\rightharpoondown_\tau SP'$, $SP'\rightharpoondown_\rho SP''$ and that  there exists a specification $SP^\tau$ over
$\Sigma'$ such that  $\mean{SP^\tau} = \Mod_{\tau}(SP)$ and $\rho$ interprets $SP^{\tau}$.
Then $SP\rightharpoondown_{\rho\circ\tau} SP''$.
\end{thm}
\begin{proof}
Since $\tau$ interprets $SP$, $\Mod_{\tau}(SP)$ also interprets $SP$.
Let $\xi\in\Ceq(\Sigma)$. Then
$SP\models \xi\; \Leftrightarrow\; \Mod_{\tau}(SP) \models \tau(\xi)
\; \Leftrightarrow\; SP^{\tau} \models \tau(\xi)$. But, since  $\rho$ interprets $SP^{\tau}$, this is equivalent to
$K \models \rho(\tau(\xi))$, for some class $K$ of $\Sigma''$-algebras.
Thus, $\rho\circ \tau$ interprets $SP$.
On the other hand, suppose $SP\models \xi$. Since $SP\rightharpoondown_\tau SP'$,
$SP'\models \tau(\xi)$. And, from $SP'\rightharpoondown_\rho SP''$ we have $SP''\models
\rho(\tau(\xi))$. Therefore, putting everything together, $SP\rightharpoondown_{\rho\circ\tau} SP''$.

\end{proof}

\subsection{Going generic}
This section introduced  refinement by interpretation in the specific  setting of algebraic specification over the institution of Horn clause logic.
Our discussion built on the fact that each specification which is flat, or equivalent to a flat specification, induces a
2-dimension deductive system corresponding to its class of models. This made possible the use of interpretations
to reason about the refinement of specifications.

A similar discussion could have been made directly over
2-dimension deductive systems. Revisiting specification $\mathrm{T}$ in example \ref{firstex} in such setting will lead
to the deductive system depicted
in Figure \ref{fig:unicads}.  On its turn, specification $\mathrm{S}$ in the same example corresponds to the
deductive system  $\mathrm{EQ}_{\Sigma}$ for the relevant signature $\Sigma$ (see Figure \ref{fig:3} in Subsection \ref{sc:eqqq}).
Note that we are now considering an arbitrary  binary predicate and not necessarily  the
equality relation. It may stand,  for example, for bisimilarity or other form of observational equivalence.
This explains the need for the explicit introduction of axioms and inference rules which would
otherwise be assumed for equality.

\begin{figure}[h]
\caixapeq{
\smallskip
\begin{center}
\begin{minipage}{0.84\textwidth}
%\begin{DSpecDefn}{DS}~
%		\item[\Sorts]   ~
%		\\              $s$
%		\item[\Ops]     ~
%		\\          $f:s \longrightarrow s$ \\
%	\item[\textbf{axioms}]  ~  $\pair{x,x}$ \\
%        \item[\textbf{inference rules}]  ~ \\
%	\\                $\displaystyle \frac{\pair{x,x'}\; \;  \pair{x',x''}}{\pair{x,x''}}$
%\hspace{7mm}   $\displaystyle \frac{\pair{x,x'} }{ \pair{f(x),f(x')}}$   \hspace{7mm} $\displaystyle \frac{\pair{x,x'} }{\pair{x',x}}$ \\
%\vspace{2mm}	\hline \vspace{2mm}%
%\end{DSpecDefn}
\begin{DSpecDefn}{DT}
			\item[\Sorts]   ~
			\\              $s$
			\item[\Ops]     ~
			\\              $ok:\longrightarrow s$
			\\              $f:s \longrightarrow s$
			\\				$test:s \* s \longrightarrow s$  \\
	\item[\textbf{axioms}]  ~\\ 
  $\mathstrut\pair{test(x,x),ok}$ \\
	\item[\textbf{inference rules}]  ~\\
	\\           \vspace{1mm}
	            $\displaystyle \frac{\pair{test(x,x'),ok} \;\; \;  \pair{test(x',x''),ok}}{\pair{test(x,x''), ok}}$    \\ \vspace{1mm}
	\\ $\displaystyle \frac{\pair{test(x,x'),ok }}{\pair{test(x',x),ok}}$   \hspace{7mm}
	 $\displaystyle \frac{\pair{test(x,x'), ok}}{\pair{test(f(x),f(x')), ok}}$    \\ \vspace{2mm}
\end{DSpecDefn}
 \end{minipage}
\end{center}
\smallskip
}
\caption{Deductive system $\mathrm{DT}$}\label{fig:unicads}
\end{figure}

\noindent
Clearly the translation
\[
 \pair{x,x'}\; \mapsto\; \pair{test(x,x'), ok} 
\]
interprets $\mathrm{EQ}_{\Sigma}$.
In the sequel, this generalisation will be carried on further, leading to a  theory of refinement by interpretation over arbitrary
$k$-dimensional deductive systems. We will resort  to a \textsc{Casl}-inspired notation to
describe (finitary) \dlogics\ (coming from algebraic specifications or not), as illustrated in Figure \ref{fig:unicads}.

%%%%%%%%%%%%%%%%%%%%%%%%%

\section{Logical  interpretation in a general setting}\label{sc:gen}
\begin{center}
\begin{minipage}{0.8\textwidth}
\emph{This section generalises translations and logical interpretations to arbitrary $k$-dimensional
deductive systems. In particular, the case of $k$-dimensional systems possessing an algebraic semantics is discussed
in some detail.  }
\end{minipage}
\end{center}

%%%%%%%%%%%%%%%%%%%%%%%%%%%%%%%%%%%%%%%%%
%\input{mmb12GEN.tex}

\subsection{Translations}
The first step to generalise \emph{refinement by interpretation} from the equational case
 is to define the notion of \emph{translation} in the general setting
of $k$-dimensional deductive systems.
The following definition generalises Definition \ref{df:eqtrans}, still assuming that, for each signature, the set
$\Va$ of variables is locally countably infinite.

\begin{defi}[Translation] \label{df:translation}
Let $\Sigma$ and $\Sigma'$ be two
signatures. A \emph{$(k,l)$-translation from $\Sigma$ to $\Sigma'$}
 is a globally finite sorted \emph{multifunction} $\mfdec{\tau}{\form{k}{\Sigma}}{\form{l}{\Sigma'}}$,
 i.e.,  for any $s \in S$ and $\bar\varphi \in \form{k}{\Sigma}_s$, $\tau_s(\bar\varphi)$ is a globally finite $S'$-sorted set of $l$-formulas over $\Sigma'$.
 \end{defi}

As before, $\tau$ is called a \emph{self translation} of $\Sigma$ whenever  $\Sigma$ and $\Sigma'$ coincide. In this
case, we say that $\tau$\emph{ commutes with substitutions} if for every substitution $\sigma$
and every formula $\bar\varphi\in \form{k}{\Sigma}$
$\tau(\sigma(\bar\varphi))=\sigma(\tau(\bar\varphi))$.
The translation can be specified by giving, for each sort $s$, the image $\tau_s(\bar x\hspace{-.09cm}:\hspace{-.09cm}s)$ for a $k$-variable $\bar x\hspace{-.09cm}:\hspace{-.09cm}s$ (see Example \ref{exe2}).
Given a $(k,l)$-translation $\tau$ and an inference rule
$\xi = \langle \Gamma,\bar\varphi\rangle$, we write $\tau(\xi)$ for the set of inference rules
$\{\langle \tau (\Gamma),\bar\psi\rangle:\bar\psi \in \tau(\bar\varphi)\}.$

A self $(k,l)$-translation $\tau$ is \emph{schematic} if there is a $S$-sorted set
$\Delta$ of $l$-formulas, where for each s, $\Delta_s(\bar x)$ is a set of $l$-formulas over $\Sigma'$  in the $k$-variable $\langle
x_0\hspace{-0.09cm}:\hspace{-0.09cm}s,\dots,x_{n-1}\hspace{-0.09cm}:\hspace{-0.09cm}s\rangle$  such that, for any $\bar \varphi \in \form{k}{\Sigma}_s$,
$\tau_s(\bar \varphi)=\Delta_s(\varphi_0,\dots, \varphi_{k-1})$.We say that a $(k,l)$-translation is \emph{functional} if the image of each $k$-formula is a singleton.
Schematic (2,2)-translations were first used in  \cite{madeira}.
Finally, the following result is the obvious generalisation of Lemma \ref{propcze} from Section \ref{sc:rbi1}.
\begin{lem}\label{propcze1}
Let $\Sigma$ be a standard signature and $\tau$ a self $(k,l)$-translation of $\Sigma$. Then
the following conditions are equivalent:
\begin{enumerate}[label=(\roman*)]
 \item \label{1} $\tau$ commutes with substitutions.
 \item \label{2} There exists a $k$-variable $\bar{x}=\langle x_0,\dots,x_{k-1}\rangle$
 and a $S$-sorted set $\Delta(\bar x)$ of $l$-formulas in $\bar{x}$
   such that, for any $\bar \varphi \in \form{k}{\Sigma}_s$, $\tau_s(\bar \varphi)=\Delta_s(\bar{\varphi})$.
\end{enumerate}

\end{lem}

\subsection{Interpretations}
Similarly to the equational case not all translations lead to refinements.
Hence, we start by generalising the definition of \emph{interpretation}.

\begin{defi}[Interpretation]Let $\tau$ be a $(k,l)$-translation from $\Sigma$ to $\Sigma'$, and
 $\cl$  a $k$-\dlogic\ over $\Sigma$. We say that $\tau$ \emph{interprets $\cl$} if there is
a $l$-\dlogic\ $\cl'$ over $\Sigma'$ such that, for any $\Gamma\cup \{\bar\varphi\} \subseteq
\form{k}{\Sigma}$, $\Gamma \vdash_\cl \bar\varphi$ if and only if $\tau(\Gamma) \vdash_{\cl'}
\tau(\bar\varphi)$. In this case we say that \emph{$\tau$ interprets $\cl$ in $\cl'$} and
\emph{$\cl'$ is a $\tau$-interpretation of $\cl$}.
\end{defi}

To illustrate this more general notion of an interpretation, consider the following examples which
capture a change
of logic paradigm. Integrating such a move in the refinement process, by witnessing refinement steps with
 this sort of interpretations, was, from the outset, the motivation for this generalisation.

\begin{exa}[$\CPC$ vs. Boolean algebras]\label{exe2}
	 The deductive system encoding the equational logic of  Boolean algebras $\cl_{\BA}$ interprets
 classical propositional logic ($\CPC$), both over the one-sorted signature
$\Sigma=\{\rightarrow,\wedge,\vee,\neg,\top,\bot\}$, under the schematic, self (1,2)-translation
$\tau(p)=\{\langle p, \top\rangle\}$.
Moreover, the
deductive system $\cl_{\HA}$, induced by the class of Heyting algebras $\HA$, also provides an interpretation of
$\CPC$ under the translation $\nu(p)=\{\langle \neg \neg p, \top\rangle\}$.
This translation also interprets $\CPC$ into $\cl_{\BA}$ which shows that
an interpretation may not be unique \cite{blok}.

Reciprocally, as one would expect,  $\CPC$ also interprets $\cl_{\BA}$, under the
(2,1)-translation $\rho(\langle p, q\rangle)=\{ p\rightarrow q, q\rightarrow p\}$.
\begin{flushright}
\vspace{-18 pt}$\Diamond$
\end{flushright}
\end{exa}

\begin{exa}
[Semilattices into posets] A semilattice can be regarded either as an
algebra or as a partial order structure. Such a duality, often useful in specifications,
can be expressed, in a natural way,  by an interpretation, actually an
equivalence between two 2-\dlogics\ over the one-sorted
signature $\Sigma=\{\wedge\}$ (see \cite{pigozzi})
depicted in Figures \ref{fig:8} and \ref{fig:9}.

\begin{figure}
\caixapeq{
\smallskip
\begin{center}
\begin{minipage}{0.8\textwidth}
\begin{DSpecDefn}{SLV}~
\item[\Enrich] $\mathrm{EQ_\Sigma}$
\item[\Axioms]     ~  \vspace{1mm}
\\              $\langle p,p\wedge p \rangle$
\\              $\langle p\wedge q,q \wedge p \rangle$
\\              $\langle p\wedge (q \wedge r),(p \wedge q) \wedge r\rangle$
\end{DSpecDefn}
 \end{minipage}
\end{center}
\smallskip
}
\caption{Semilattices as algebras.}\label{fig:8}
\end{figure}

\begin{figure}
\caixapeq{
\smallskip
\begin{center}
\begin{minipage}{0.8\textwidth}
\begin{DSpecDefn}{SLP}~
%\item[\Enrich] $\mathrm{EQ_\Sigma}$
\item[\Axioms]     ~ \vspace{1mm}
\\              $\langle p, p \rangle$
\\              $\langle p,p \wedge p \rangle$
\\              $\langle p \wedge q,p \rangle$
\\				$\langle p\wedge q,q \rangle$\vspace{1mm}
\item[\textbf{inference rules}]     ~  \vspace{2mm}
\\            $\displaystyle\frac{\langle x, y\rangle,\langle y, z\rangle}{\langle x, z \rangle}$
\\
\\              $\displaystyle\frac{\langle x_0, y_0\rangle, \langle x_{1},y_{1}\rangle}{(x_0 \wedge
x_{1},y_0\wedge y_{1}\rangle}$
\end{DSpecDefn}
 \end{minipage}
\end{center}
\smallskip
}
\caption{Semilattices as order structures.}\label{fig:9}
\end{figure}

\noindent
The schematic translation $\tau$, defined by the multifunction $\tau(\langle p,q\rangle)=\{\langle p,q\rangle, \langle q,p
\rangle\}$, witnesses the interpretation of  $\mathrm{SLV}$ by $\mathrm{SLP}$.
\begin{flushright}
\vspace{-18 pt}$\Diamond$
\end{flushright}
\end{exa}

Clearly, interpretations compose in the sense that if $\cl'$ is a
$\tau$-interpretation of $\cl$ and $\cl''$ is a $\rho$-interpretation of $\cl'$ then
$\rho\circ\tau$ interprets $\cl$ in $\cl''$. Other properties need further investigation.
The following subsection explores a class of
interpretations specifically relevant for software design.

\subsection{Towards an algebraic semantics}\label{sc:algsem}

There are $k$-\dlogics\ to which an algebraic specification can be associated, thus
providing  an alternative semantics (called \emph{algebraic semantics} in the context of
algebraic logic).
It is well known \cite{blok} that this association is not unique and may not exist.
Our investigation starts with the following definition which generalises Definition \ref{df:tau:model}:

\begin{defi}[$\tau$-model] Let $\tau$ be a $(k,l)$-translation from $\Sigma$
to $\Sigma'$ and $\cl$ a $k$-\dlogic\ over $\Sigma$. An $l$-structure $\ca$ is a
\emph{$\tau$-model of} $\cl$ if for any  $\Gamma\cup \{\bar\varphi\} \subseteq
\form{k}{\Sigma}$, $\Gamma \vdash_\cl \bar\varphi$  implies $ \tau(\Gamma) \models_\ca
\tau(\bar\varphi)$. The class of all $\tau$-models of $\cl$, denoted by $\Mod_\tau (\cl)$, is
called the \emph{$\tau$-model class of $\cl$}.
\end{defi}

As mentioned above, the semantic consequence associated to a class of $k$-structures
defined over $\form{k}{\Sigma}$, is always a $k$-\dlogic\, even if it fails to be
specifiable. Hence, $\models_{\Mod_\tau (\cl)}$ is a \dlogic\ which we will denote by $\cl^\tau$.
Furthermore,

\begin{thm}\label{axi00} Let $\tau$ be a $(k,l)$-translation from $\Sigma$
to $\Sigma'$ and $\cl$ a $k$-\dlogic\ over $\Sigma$. If $\tau$ interprets $\cl$, then the
$l$-\dlogic\ $\cl^\tau$ is a $\tau$-interpretation of $\cl$; moreover, this is the $\tau$-interpretation of $\cl$ with
the largest class of models.
\end{thm}

\begin{proof}Suppose that $\tau$ interprets $\cl$. Let $\cl'$ be a specification which is
a $\tau$-interpretation of $\cl$. Then for any $\Gamma\cup \{\bar\varphi\} \subseteq
\form{k}{\Sigma}$, $\Gamma \vdash_\cl \bar\varphi$ iff $\tau(\Gamma) \vdash_{\cl'}
\tau(\bar\varphi)$ iff $\tau(\Gamma) \models_{\Mod(\cl')} \tau(\bar\varphi)$. Hence all
models of $\cl'$ are $\tau$-models of $\cl$. Thus, $\Mod(\cl') \subseteq \Mod_\tau (\cl)$.
\\So, it is enough to prove that $\cl^\tau$ is a $\tau$-interpretation of $\cl$. Let
$\Gamma\cup \{\bar\varphi\} \subseteq \form{k}{\Sigma}$. It is clear that $\Gamma \vdash_\cl
\bar\varphi$ implies $\tau(\Gamma) \vdash_{\cl^\tau} \tau(\bar\varphi)$. Suppose now that
$\tau(\Gamma) \vdash_{\cl^\tau} \tau(\bar\varphi)$. Let  $\cl'$ be a specification that is a
$\tau$-interpretation of $\cl$ (it exists since $\tau$ interprets $\cl$). Since, $\Mod(\cl')
\subseteq \Mod_\tau (\cl) $, $\tau(\Gamma) \vdash_{\cl'} \tau(\bar\varphi)$. Thus $\Gamma
\vdash_\cl \bar\varphi$ because $\cl'$ is a $\tau$-interpretation of $\cl$.
\end{proof}

Our focus on algebraic specifications entails the need for paying special attention
to $(k,2)$-translations, which, mapping  $k$-formulas to $2$-formulas, provide
a  way to relate an arbitrary  $k$-\dlogic\  to a suitable class of algebras.
For the remaining of this section we will consider a $2$-formula $\langle t, t'\rangle$ as an equation
$t \approx t'$.

Let $\tau$ be a $(k,2)$-translation from $\Sigma$ to $\Sigma'$ and $\cl$ a $k$-\dlogic. A class
$K$ of $\Sigma$-algebras is said to be a \emph{$\tau$-algebraic semantics of $\cl$} if $\tau$
interprets $\cl$ in $\models_K$. Thus, we define the algebraic
model class $K_\cl^\tau$ over $\Sigma'$ as the class of
algebraic reducts of the $\tau$-models  $\cl$ taking the identity as a filter.
Formally,
\begin{equation*}
K_\cl^\tau\;  =\;  \{ A| \langle A, \triangle_A\rangle \text{ is a } \tau\text{-model}  \},
\end{equation*}

\noindent which paves the way to the following corollary:
\begin{cor}\label{cor:largest}
Given a $(k,2)$-translation $\tau$ from $\Sigma$ to $\Sigma'$,
 and a $k$-\dlogic\ $\cl$ over $\Sigma$, if there is a $\tau$-algebraic semantics of $\cl$,
  then the class $K_\cl^\tau$ is the largest $\tau$-algebraic
semantics of $\cl$, \emph{i.e.}, with the largest class of models.
Moreover, $K_\cl^\tau$ is finitely axiomatised whenever $\cl$ is finitely axiomatisable.
\end{cor}

%-----------------------
In practice, however, it may happen that $K_\cl^\tau$ is too wide for the envisaged purposes,
namely to discuss implementations.
The following theorem gives a sufficient and necessary
condition for a subclass of $K_\cl^\tau$ to be a $\tau$-algebraic semantics of $\cl$.
Similar results are well known for sentential logics \cite{blok}. In this paper, however, we
reformulate them for $k$-dimensional and many sorted logics, since they are a vehicle to
sufficient and necessary conditions for a deductive system to have an algebraic semantics.
Consider, therefore, the mapping $\taupe:\thy (\cl)\rightarrow\thy (K)$  defined by
$\taupe(T)=\Cn_{K}(\tau(T))$, for all $T\in\thy(\cl)$. Then,

\begin{lem}\label{extalsem1}
Let $\cl$ be a \dlogic, $\tau$  a self $(k,2)$-translation of $\Sigma$ which commutes with
substitutions and $K\subseteq K_\cl^\tau$. The following conditions are equivalent:
    \begin{enumerate}[label=(\roman*)]
        \item $K$ is a $\tau$-algebraic semantics of $\cl$.
        \item $\taupe$ is injective.
    \end{enumerate}
\end{lem}

\begin{proof}
Let $T_1,T_2\in\thy (\cl)$ and $\bar \alpha\in T_1$. Suppose $\taupe(T_1)=\taupe(T_2)$. We
have that $\tau(\bar \alpha)\subseteq\tau(T_1)\subseteq\taupe(T_1)=\taupe(T_2)$, \emph{i.e.},
$\tau(T_2)\models_K\tau(\bar \alpha)$. Since $K$ is an $\tau$-algebraic semantics of $\cl$, $T_2\vdash_\cl\bar \alpha$, \emph{i.e.}, $\bar\alpha\in T_2$. Thus $T_1\subseteq T_2$.
Similarly, we can prove that $T_2\subseteq T_1$. We conclude that $\taupe$ is injective.

Conversely, let $\Gamma\cup\{\bar\alpha\}\subseteq \form{k}{\Sigma}$. Since $K$ is a class
of algebraic reducts of  $\tau$-models of $\cl$, we have that $\Gamma\vdash_\cl\bar\alpha$
implies $\tau(\Gamma)\models_K\tau(\bar\alpha)$. Now, suppose
$\tau(\Gamma)\models_K\tau(\bar\alpha)$. Thus,
$\Cn_K(\tau(\Gamma))=\Cn_K(\tau(\Gamma\cup\{\bar\alpha\}))$. Since
$\Gamma\subseteq\Cn_\cl(\Gamma)$, we have that $\tau(\Gamma)\subseteq\tau(\Cn_\cl(\Gamma))$.
Thus $\Cn_K(\tau(\Gamma))\subseteq \Cn_K(\tau(\Cn_{\cl}(\Gamma)))=\taupe(\Cn_\cl(\Gamma))$. To
prove the reverse inclusion, let $t\approx t'\in\taupe(\Cn_\cl(\Gamma))$.%, i.e.,
%$\tau[\Cn_S(\Gamma)]\ecr\alpha\approx\beta$.
 Thus
$\{\tau(\xi):\Gamma\vdash_\cl\xi\}\models_K\et$. Again, since $K$ is a class  of algebraic
reducts of  $\tau$-models of $\cl$, for all $\et\in \form{2}{\Sigma}$, we have that
$\Gamma\vdash_\cl\xi$ implies $\tau(\Gamma)\models_K\tau(\xi)$. Hence
$\tau(\Gamma)\models_K\et$, i.e., $\et\in \Cn_K(\tau(\Gamma))$. Therefore, for all
$\Gamma\subseteq \form{k}{\Sigma}$, $\taupe(\Cn_\cl(\Gamma))=\Cn_K(\tau(\Gamma))$. Thus, $\taupe(\Cn_\cl(\Gamma))=\taupe(\Cn_{\cl}(\Gamma\cup\{\bar\alpha\}))$.
Since $\taupe$ is injective, $\Cn_\cl(\Gamma)=\Cn_\cl(\Gamma\cup\{\bar\alpha\})$, i.e.,
$\Gamma\vdash_\cl\bar\alpha$.
\end{proof}

Lemma \ref{extalsem1} and the fact that class $K_\cl^\tau$ is a $\tau$-algebraic semantics,
entail another important result: if $\cl$ has a $\tau$-algebraic semantics, then any extension of $\cl$ also
has a $\tau$-algebraic semantics, for $\tau$  a self $(k,2)$-translation of $\Sigma$ commuting with
substitutions. This is recorded in Theorem \ref{extalsem} below, whose proof requires the
following lemma.

\begin{lem}\label{extalsem2} Let $\cl$ be a specifiable $k$-\dlogic\ and $\tau$ a self
$(k,2)$-translation of $\Sigma$ which commutes with substitutions. Suppose $K=K_\cl^\tau$ is an
$\tau$-algebraic semantics of $\cl$. If $\cl'$ is an extension of $\cl$, and
$K^\prime=K_\cl^\tau$ then $\tau_{\cl^\prime,K^\prime}$ equals $\taupe$ restricted
to $\thy (\cl^\prime)$.
\end{lem}
\begin{proof}
Let $T\in\thy(\cl^\prime)$. Since $\cl^\prime$ is an extension of $\cl$, $K^\prime\subseteq K$,
and $\models_{K^\prime}$ is an extension of $\models_K$. Hence $\Cn_K(\tau[T])\subseteq
\Cn_{K^\prime}(\tau[T])$, \emph{i.e.}, $\taupe[T]\subseteq\tau_{\cl^\prime,K^\prime}[T]$. For the
reverse inclusion note that $K^\prime$ can be axiomatised by a set of
axioms and a set of inference rules. It is not difficult to see that $\Cn_K(\tau[T])$ contains
all substitution instances of the axioms of $\models_{K^\prime}$ and is closed under the
inference rules of $\models_{K^\prime}$.
Let $\alpha\in\Thm(\cl^\prime)$
and $e$ be a substitution. Since $\theo(\cl^\prime)$ is closed under substitutions,
$\vdash_{\cl^\prime}e(\alpha)$. Thus for all $T\in\Th(\cl^\prime)$, we have that $e(\alpha)\in
T$. As $\tau$ commutes with arbitrary substitutions,
$e[\tau(\alpha)]=\tau[e(\alpha)]\subseteq\tau[T]\subseteq \Cn_K(\tau[T])$. Thus
$\Cn_K(\tau[T])$ contains all substitution instances of axioms of $\models_{K^\prime}$. Now, let
$\{\alpha_i:i<n\}\vdash_{\cl^\prime}\beta$ be an inference rule of $\cl^\prime$ and $e$ a
substitution such that $\{e[\tau(\alpha_i)]:i<n\}\subseteq \Cn_K(\tau[T])$, i.e.,
$\tau[T]\models_K e[\tau(\alpha_i)]$ for all $i<n$. Since $\tau$ commutes with arbitrary
substitutions, $e[\tau(\alpha_i)]=\tau[e(\alpha_i)]$ for all $i<n$, \emph{i.e.},
$\tau[T]\models_K\tau[e(\alpha_i)]$, for all $i<n$. As $K$ is an algebraic semantics of $\cl$,
it follows that $T\vdash_S e(\alpha_i)$ for all $i<n$, \emph{i.e.}, $\{e(\alpha_i):i<n\}\subseteq T$.
By structurality of $\cl^\prime$, $\{e(\alpha_i):i<n\}\vdash_{\cl^\prime}e(\beta)$. Since $T\in
\Th(\cl^\prime)$, we have that $e(\beta)\in T$. Thus
$e[\tau(\beta)]=\tau[e(\beta)]\subseteq\tau[T]\subseteq \Cn_K(\tau[T])$. Therefore
$\Cn_K(\tau[T])$ is closed under the inference rules of $K^\prime$. By the characterisation of
a theory in a deductive system, we have proved that $\Cn_\Kal(\tau[T])\in\thy(\Kal^\prime)$.
Since $\tau[T]\subseteq\Cn_K(\tau[T])$ and $\Cn_{K^\prime}(\tau[T])$ is the least
$\cl^\prime$-theory that contains $\tau[T]$, we have that $\Cn_{K^\prime}(\tau[T])\subseteq
\Cn_K(\tau[T])$, i.e., $\tau_{\cl^\prime,K^\prime}[T]\subseteq\taupe[T]$.
\end{proof}

\noindent
The following main result can now be proved:
\begin{thm}\label{extalsem}
Let $\cl$ be a specifiable $k$-\dlogic\ and $\tau$  a self $(k,2)$-translation of $\Sigma$ which commutes with
substitutions. If $\cl$ has an $\tau$-algebraic semantics, then any extension of $\cl$
has a $\tau$-algebraic semantics as well.
\end{thm}
\begin{proof}
Let  $K$ be a $\tau$-algebraic semantics of $\cl$. Let $\cl^\prime$ be an extension of $\cl$
and $K^\prime=K_{\cl'}^\tau$. By Corollary \ref{cor:largest}, the class
$K_\cl^\tau$ is a $\tau$-algebraic semantics. Since, by Lemma \ref{extalsem1}, the
mapping $\taupe$ is injective, we have, by Lemma \ref{extalsem2}, that the mapping
$\tau_{\cl^\prime,K^\prime}$ is also injective. Again by Lemma \ref{extalsem1},
$K^\prime$ is a $\tau$-algebraic semantics of $\cl^\prime$.
\end{proof}

\section{Refinement by interpretation: The general case}\label{sc:rbi2}
\begin{center}
\begin{minipage}{0.8\textwidth}
\emph{This section revisits the notion of refinement by interpretation in the
general setting of  arbitrary $k$-dimensional deductive systems, and illustrates the design flexibility it entails.
}
\end{minipage}
\end{center}
~\\

%%%%%%%%%%%%%%%%

As before, our concern is to put forward a precise, but flexible notion of what counts for a valid refinement step
in software design.
Our starting point is the following syntactic-grounded notion,

\begin{defi}[(Logical) refinement] Let $\Sigma$ and $\Sigma'$ be two signatures such that
$\Sigma\subseteq\Sigma'$, and
  $\mathcal{L}, \mathcal{L'}$
two $k$-\dlogics\ over $\Sigma$ and $\Sigma'$, respectively. We say that $\mathcal{L'}$ is a
\emph{(logical) refinement of} $\mathcal{L}$, in symbols $\mathcal{L}\rightsquigarrow
\mathcal{L'}$, if for any
 $\Gamma\cup \{\bar\varphi\} \subseteq
\form{k} {\Sigma}$, $$\Gamma\vdash_\cl \bar \varphi \Rightarrow \Gamma\vdash_{\cl'} \bar
\varphi.$$
\end{defi}

Note that, when $\mathcal{L}$ is specifiable, $\mathcal{L}\rightsquigarrow \mathcal{L'}$ if all
the axioms of $\cl$ are theorems of $\mathcal{L}'$ and the theories of $\cl'$ are compatible
with the inference rules of $\cl$.

 \begin{exa}
Modal logic  $\mathrm{S5^G}$ forms a (logical) refinement of $\CPC$. Consider the modal signature
$\Sigma=\{\rightarrow,\wedge,\vee,\neg,\top,\bot,\Box\}$. Modal logic $\mathrm{K}$ is obtained from
 $\CPC$ by adding the symbol $\Box$ to the signature, the axiom $\Box\,(p\rightarrow q)\rightarrow (\Box\, p
\rightarrow\Box\, q)$ and the inference rule $\displaystyle{\frac{p}{\Box\, p}}$. Logic $\mathrm{S5^G}$, on the other hand,
enriches the signature of
$\mathrm{K}$ with the symbol $\Diamond$, and  $\mathrm{K}$ itself with the
axioms $\Box\,p\rightarrow p$, $\Box\,p \rightarrow \Box\Box\,p$ and $\Diamond\,p\rightarrow
\Box\Diamond\,p$ \cite{pigozzi}. Hence, since the signature of both systems contains the
signature of $\CPC$ and their presentations result from the introduction of extra
 axioms and inference rules to the $\CPC$ presentation, we have, by the previous fact that
$\CPC\rightsquigarrow \mathrm{K}$ and $\CPC \rightsquigarrow \mathrm{S5^G}$ (actually, $\CPC\rightsquigarrow \mathrm{K}
\rightsquigarrow \mathrm{S5^G}$). Thus, refining $\CPC$ in this way, we acquire enough expressivity  to state
properties over propositions like \emph{it is necessary that $\phi$} (by $\Box \, \phi$) and
\emph{it is possible that $\phi$} (by $\Diamond \, \phi$). This kind of refinement makes possible the accommodation of a new type of requirements, modally expressed, along the refinement process.
\begin{flushright}
\vspace{-18 pt}$\Diamond$
\end{flushright}
\end{exa}

\begin{thm}\label{ref_syn_model}  Let $\Sigma$ be a signature  and  $\mathcal{L}$ and $\mathcal{L'}$ two $k$-\dlogics\ over $\Sigma$. Then the following conditions are equivalent
  \begin{enumerate}[label=(\roman*)]
    \item $\mathcal{L}\rightsquigarrow \mathcal{L'}$
    \item $\Mod(\mathcal{L'})\subseteq \Mod(\mathcal{L}).$
    \end{enumerate}
\end{thm}
\begin{proof}
(i) $\Rightarrow$ (ii). Suppose $\mathcal{L}\rightsquigarrow \mathcal{L'}$. Let $\ca\in\Mod(\cl')$ and $\Gamma\cup\{\bar \varphi\}\subseteq \form{k}{\Sigma}$. Suppose $\Gamma \vdash_{\mathcal{L}}\bar\varphi$. We have by (i) that $\Gamma \vdash_{\cl'}\bar\varphi$ and hence
$\Gamma \vdash_\ca\bar\varphi$. Therefore $\ca \in \Mod(\mathcal{L})$.

\noindent ii) $\Rightarrow$ i). Suppose $\Mod(\mathcal{L'})\subseteq \Mod(\mathcal{L}).$ Let $\Gamma\cup\{\bar \varphi\}\subseteq \form{k}{\Sigma}$. Suppose $\Gamma \vdash_{\mathcal{L}}\bar\varphi$. Let $\ca\in\Mod(\cl')$. By ii) we have $\ca\in\Mod(\cl)$ and hence $\Gamma\vdash_{\ca}\bar\varphi$. Therefore, by Completeness, $\Gamma\vdash_{\cl'}\bar\varphi$.
\end{proof}

A coarser and more flexible definition of refinement, however, is provided by the notion of logical
interpretation, as already shown in the equational case. Formally,

\begin{defi}[Refinement by interpretation]
Let $\cl$ be a $k$-\dlogic\ over $\Sigma$ and $\tau$ a $(k,l)$-translation from $\Sigma$ to
$\Sigma'$, which interprets $\cl$.
We say that a $l$-\dlogic\ $\cl'$ over $\Sigma'$ \emph{refines
the \dlogic\ $\cl$ via the interpretation $\tau$}, in symbols $\cl\rightharpoondown_\tau \cl'$,
if for any
 $\Gamma\cup \{\bar\varphi\} \subseteq
\form{k}{\Sigma}$,
$$ \Gamma \vdash_\cl \bar\varphi \; \imp\; \tau(\Gamma) \vdash_{\cl'} \tau(\bar\varphi). $$
\end{defi}

The requirement that $\tau$ has to interpret $\cl$ is necessary in order to enforce some control
over the class of models of the \dlogic\ $\cl'$. In particular, this guarantees that $\Mod(\cl')$ has to be smaller than $\Mod_{\tau}(\cl)$.

The following two examples illustrate this general notion of refinement
 at work.

\begin{exa}
Any subclass of the class of Boolean algebras induces a refinement by interpretation of $\CPC$,  based on
 the usual (1,2)-translation $\tau$ given by $\tau(p)=\{\langle p,\top\rangle \})$.
\begin{flushright}
\vspace{-18 pt}$\Diamond$
\end{flushright}
\end{exa}

\begin{exa}\label{ex:ordbams}
Consider the fragment of the specification BAMS, of a toy bank account management system,
 given in Example \ref{ex:bams1}, regarded as a $2$-\dlogic.
Suppose  we intend to refine this system by imposing that the \emph{balance of each
account has to  be positive}. This cannot be easily expressed in  (strict) equational
logic. However, it can be captured as a refinement by interpretation.
Actually, consider  the $2$-\dlogic\ over $\Sigma$ sketched in Figure \ref{ordbams},  in which $n, n', n''$ are variables of sort $Int$,
$x,y$ of sort $Sys$ and $i,j$ of sort $Ac$.
Notice that only a few  axioms and inference rules are shown for illustration purposes.
Intuitively we intend to interpret differently the
binary predicates: as equality for the carriers of $Ac$ and $Sys$; as  $\leq$ for the integers.

\begin{figure}
\caixapeq{
%\smallskip
\begin{center}
\begin{minipage}{0.85\textwidth}
\begin{DSpecDefn}{ORDBAMS}
\item[\Axioms]  ~
\\              $ \langle n , n \rangle$\;
                $ \langle i , i \rangle$\;
                 $ \langle x , x \rangle$
\\              $ \langle n, n+0 \rangle$
\\                $ \langle n+0 , n \rangle$
\\                $ \langle n+n' , n'+n \rangle$
\\              $\dots$
\\              $\langle \balance(\dep(x,i,n),i) \; ,\; \balance(x,i) +
n \rangle$
\\            $\langle \balance(x,i) + n \; ,\; \balance(\dep(x,i,n),i)\rangle $
\\       $ \langle \balance(\withd(x,i,n),i) \; , \; \max(\balance(x,i) - n, 0)\rangle$
\\              $ \langle \max(\balance(x,i) - n, 0) \; , \;\balance(\withd(x,i,n),i) \rangle$
\\              $\cdots$
\item[\textbf{inference rules}]  ~\\
\\ $\displaystyle\frac{ \langle n, n'\rangle \; \langle n',  n \rangle }{ \langle
  n, n'' \rangle}$ \vspace{2.5mm}
\\              $\displaystyle\frac{ \langle n,n'\rangle}{ \langle
  n+n'', n'+n'' \rangle}$
\;               $\displaystyle\frac{ \langle n,n'\rangle}{ \langle
  -n', -n \rangle}$ \; $\displaystyle\frac{ \langle n,n' \rangle}{ \langle
  s(n), s(n') \rangle}$ \vspace{2.5mm}
  \\
   $\displaystyle \frac{\langle x,y \rangle \; \langle i,j \rangle}{\langle \bal(x,i),\bal(y,j) \rangle} \; \;
   \displaystyle \frac{\langle x,y \rangle \; \langle i,j \rangle}{\langle \bal(y,j),\bal(x,i) \rangle} $ \vspace{2.5mm}
\\ $\displaystyle \frac{\langle x,y \rangle   \; \langle i,j \rangle  \; \langle n,n' \rangle \; \langle n',n \rangle}{\langle \dep(x,i,n),\dep(y,j,n') \rangle}  \vspace{2.5mm}
\\
\displaystyle \frac{\langle x,y \rangle \;  \langle i,j \rangle  \; \langle n,n' \rangle \; \langle n',n \rangle }{\langle \withd(x,i,n),\withd(y,i,n) \rangle}$ \vspace{2.5mm}
\\
              $\cdots$\\
              $\displaystyle\frac{\langle x, y \rangle}{\langle y, x \rangle}$
%              %
%              \item[\textbf{inference rules}]  ~\\
%\\ $\displaystyle\frac{ \langle n\:Int \;,\;n'\:Int\rangle \; \langle n'\:Int \;,\;n''\:Int\rangle }{ \langle
%  n, n'' \rangle}$
%\\              $\displaystyle\frac{ \langle n\:Int \;,\;n'\:Int\rangle}{ \langle
%  n+n'', n'+n'' \rangle}$
%\;               $\displaystyle\frac{ \langle n\:Int \;,\;n'\:Int\rangle}{ \langle
%  -n', -n \rangle}$ \; $\displaystyle\frac{ \langle n\:Int \;,\;n'\:Int\rangle}{ \langle
%  s(n), s(n') \rangle}$\\
%
%   $\displaystyle \frac{\langle x\:Ac,y\:Ac \rangle }{\langle \bal(x)\:Int,\bal(y)\:Int \rangle}; \; \;  \displaystyle \frac{\langle x\:Ac,y\:Ac \rangle }{\langle \bal(y)\:Int,\bal(x)\:Int \rangle}; $
%\\ $\displaystyle \frac{\langle x\:Ac,y\:Ac \rangle }{\langle \dep(x)\:Ac,\dep(y)\:Ac \rangle}; \; \; \displaystyle \frac{\langle x\:Ac,y\:Ac \rangle }{\langle \withd(x)\:Ac,\withd(y)\:Ac \rangle}$
%
%
%              $\cdots$
%
%              $\displaystyle\frac{\langle x\hspace{-.09cm}:\hspace{-.09cm}Ac, y\hspace{-.09cm}:\hspace{-.09cm}Ac\rangle}{\langle y\hspace{-.09cm}:\hspace{-.09cm}Ac, x\hspace{-.09cm}:\hspace{-.09cm}Ac \rangle};$
%              %
\end{DSpecDefn}
\end{minipage}
\end{center}
}
\caption{Revisiting BAMS.} \label{ordbams}
\end{figure}
%}

\medskip

\noindent
Let us take a fixed-semantics approach by
fixing the $Int$ component as the integers endowed with the usual operations. Consider, thus, the
following subclasses of the model class of $\mathrm{BAMS}$ and $\mathrm{ORDBAMS}$, respectively:

\begin{align*}
\mathrm{BAMS}^{\mathbb{Z}}\; =& \; \{\big\langle A, F
\big\rangle \in  \Mod(\mathrm{BAMS}):
A_{Int}=\mathbb{Z} \, \& \,F_{Int}=id_{\mathbb{Z}}\}\\
\mathrm{ORDBAMS}^{\mathbb{G}} \; =& \;
\{\big\langle A, G \big\rangle \in \Mod(\mathrm{ORDBAMS}):
A_{Int}=\mathbb{Z} \,\& \,G_{Int}={}\leq\}
\end{align*}
%\medskip

\noindent
Notice that a structure $\langle A,F\rangle$ reduces to an algebra when, for each sort, the corresponding filter in $F$ is the
identity.
Let  $\tau$ be the
(2,2)-translation defined schematically by
\begin{align*}
\tau_{Int}(\langle n,n'\rangle)\; =& \; \{\langle n,n'\rangle,\langle n',n\rangle \}\\
\tau_{Ac}(\langle i,j\rangle)\; =& \; \{\langle i,j\rangle\}\\
\tau_{Sys}(\langle x,y\rangle)\; =& \; \{\langle x,y\rangle\}
\end{align*}
The underlying intuition
 is that an equation $n\approx n'$ of sort $Int$ is translated into two
inequalities $n\leq n'$ and $n \leq n'$.
Clearly, $\models_{\mathrm{ORDBAMS}^{\mathbb{Z}}}$ is a $\tau$-interpretation of
$\models_{\mathrm{BAMS}^\mathbb{Z}}$.
 Therefore,  the deductive system which extends $\mathrm{ORDBAMS}$ to
 capture the extra requirement $\langle 0, \bal(x,i)\rangle$ is obtained
 as a $\tau$-refinement of the original one.
 \begin{flushright}
\vspace{-18 pt}$\Diamond$
\end{flushright}
\end{exa}

\CUT{

\begin{exa}\label{ex:ordbams}
Consider the fragment of the specification BaMS, of a toy bank account management system,
 given in Example \ref{ex:bams1}.
Suppose  we intend to refine this system by imposing that the \emph{balance of each
account has to  be positive}. Naturally, this cannot be expressed in a (strict) equational
logic. However, it can be captured as a refinement by interpretation.
Actually, consider  the following $2$-\dlogic\ over $\Sigma$ depicted in Figure \ref{fig:10}

\begin{figure}
\caixapeq{
\smallskip
\begin{center}
\begin{minipage}{0.8\textwidth}
\begin{SpecDefn}{ORDBAMS}
\item[\Enrich]  $\mathrm{INT} + \mathrm{BAMS}$
\item[\Axioms]     ~
\\
\\              $\langle x\hspace{-.09cm}:\hspace{-.09cm}s,x\hspace{-.09cm}:\hspace{-.09cm}s\rangle, s \in \{Ac, Int\};$
\\              $\langle \bal(\dep(x,n)), \bal(x)+n\rangle;$
\\             $ \langle  \bal(x)+n,\bal(\dep(x,n))\rangle;$
\\              $\langle \bal(\withd(x,n)), \bal(x)+ (-n)\rangle;$
\\             $\langle \bal(x)+ (-n),\bal(\withd(x,n))\rangle;$
\\
\item[\textbf{inference rules}]  ~
\\
\\              $\displaystyle\frac{\langle x\hspace{-.09cm}:\hspace{-.09cm}Ac, y\hspace{-.09cm}:\hspace{-.09cm}Ac\rangle}{\langle y\hspace{-.09cm}:\hspace{-.09cm}Ac, x\hspace{-.09cm}:\hspace{-.09cm}Ac \rangle};$
\\
\\              $\displaystyle \frac{\langle x,y \rangle ; \langle w,z\rangle}{\langle x+w,y+z \rangle}; \; \; \displaystyle \frac{\langle x,y \rangle }{\langle -y,-x \rangle}$
\\
\\              $\displaystyle \frac{\langle x\hspace{-.09cm}:\hspace{-.09cm}s,y\hspace{-.09cm}:\hspace{-.09cm}s \rangle
\; \; \; \; \; \langle y\hspace{-.09cm}:\hspace{-.09cm}s,z\hspace{-.09cm}:\hspace{-.09cm}s\rangle}{\langle x\hspace{-.09cm}:\hspace{-.09cm}s,z\hspace{-.09cm}:\hspace{-.09cm}s \rangle},\; \;  s \in \{Ac, Int\};
$
\\
\\              $\displaystyle \frac{\langle x,y \rangle }{\langle \bal(x),\bal(y) \rangle}; \; \;  \displaystyle \frac{\langle x,y \rangle }{\langle \bal(y),\bal(x) \rangle}; $
\\ $\displaystyle \frac{\langle x,y \rangle }{\langle \dep(x),\dep(y) \rangle}; \; \; \displaystyle \frac{\langle x,y \rangle }{\langle \withd(x),\withd(y) \rangle}$
\end{SpecDefn}
 \end{minipage}
\end{center}
\smallskip
}
\caption{The bank system revisited}\label{fig:10}
\end{figure}

\noindent
Let us take a fixed-semantics approach by
fixing the $Int$ component as the integers, endowed with the usual operations,
and the corresponding component of their filters as the identity relation. Consider, then, the
following subclasses of model class of the above $2$-\dlogic:\\

%\medskip

\begin{align*}
\mathrm{BAMS}^{\mathbb{Z}}\; =& \; \{\big\langle\langle A_{Ac},A_{Int}\rangle,\langle
F_1,F_2\rangle\big\rangle \in  \Mod(\mathrm{BAMS}):
A_{Int}=\mathbb{Z} \, \& \,F_2=id_{\mathbb{Z}}\}\\
\mathrm{ORDBAMS}^{\mathbb{Z}} \; =& \;
\{\big\langle\langle A_{Ac},A_{Int}\rangle,\langle
G_1,G_2\rangle\big\rangle \in \Mod(\mathrm{ORDBAMS}):
A_{Int}=\mathbb{Z} \,\& \,G_2=\leq\}
\end{align*}
%\medskip

\noindent Let  $\tau$ be the
(2,2)-translation defined schematically by
\begin{align*}
\tau_{Int}(\langle x,y\rangle)\; =& \; \{\langle x,y\rangle,\langle y,x\rangle \}\\
\tau_{Ac}(\langle x,y\rangle)\; =& \; \{\langle x,y\rangle\}
\end{align*}
The underlying intuition
 is that an equation $x\approx y$ of sort $Int$ is translated into two
inequalities $x\leq y$ and $y\leq x$.
Clearly, $\models_{\mathrm{ORDBAMS}^{\mathbb{Z}}}$ is a $\tau$-interpretation of
$\models_{\mathrm{BAMS}^\mathbb{Z}}$.
 Therefore,  adding  axiom $\langle 0, \bal(x)\rangle$ the envisaged specification is obtained
 as a $\tau$-refinement of the original one.
 \begin{flushright}
\vspace{-18 pt}$\Diamond$
\end{flushright}
\end{exa}

}

We complete the discussion of refinement by interpretation in this general setting by establishing
its connection to classical, signature morphism based refinement. What follows lifts the
corresponding discussion in subsection \ref{sc:mainth} to the level of $k$-\dlogics\ and their interpretations:

\begin{thm} \label{carct_ref_int}Let $\cl$ and $\cl'$ be a $k$-\dlogic\ over $\Sigma$ and $l$-\dlogic\ over $\Sigma'$ respectively. Let
 $\tau$ be a $(k,l)$-translation from
$\Sigma$ to $\Sigma'$. Then the following conditions are equivalent
    \begin{enumerate}[label=(\roman*)]
    \item $\cl\rightharpoondown_\tau \cl'$;
    \item $\cl'$ is a refinement of some $ \tau$-interpretation of
      $\cl$ (i.e.,
    there is a $l$-\dlogic\ $\cl^0$ which $\tau$-interprets $\cl$ and $\cl^0 \rightsquigarrow  \cl'$).
    \end{enumerate}
\end{thm}
\begin{proof}
Suppose $\cl\rightharpoondown_\tau \cl'$. Then, by Theorem \ref{axi00}, $\Mod(\cl')$ is a
subclass of the class of $\tau$-models of
$\cl$, \ie,  $\Mod(\cl') \subseteq \Mod(\cl^\tau)$. Therefore, by Theorem
\ref{ref_syn_model}, $\cl^\tau \rightsquigarrow \cl'$. So, condition (ii) holds for
$\cl^0=\cl^\tau$.

Suppose now there is a $l$-\dlogic\ $\cl^0$ which $\tau$-interprets $\cl$ and $\cl^0
\rightsquigarrow \cl'$. Let $\Gamma\cup\{\bar\varphi\}\subseteq\form{k}{\Sigma}$. Then
$$\Gamma \vdash_\cl \bar\varphi\Leftrightarrow\tau(\Gamma) \vdash_{\cl^0}
\tau(\bar\varphi)\Rightarrow \tau(\Gamma) \vdash_{\cl'} \tau(\bar\varphi).$$ \noindent The
equivalence holds since $\tau$ interprets $\cl$ in $\cl^0$. The implication holds since
$\cl^\tau \rightsquigarrow \cl'$. Therefore, $\cl\rightharpoondown_\tau \cl'$.
 \end{proof}

\begin{exa}  \label{ex:tool}
Suppose a requirements specification is provided in  $\CPC$,
but an implementation is sought in which the system properties are expected
to be shown in a constructive way resorting, for example, to a theorem prover. This
entails the need for  refactoring the specification to some variant of intuitionist logic.
Based on Theorem \ref{carct_ref_int}  we have $\CPC \rightharpoondown_\tau \HA \rightharpoondown_\rho \mathrm{IPC}$, with
$\tau(p)=\{\langle\neg\neg p,\top \rangle\}$  and get $\rho(\langle p,q\rangle)=\{p\rightarrow q,
q\rightarrow p\}$ doing the job.
\begin{flushright}
\vspace{-18 pt}$\Diamond$
\end{flushright}
\end{exa}

The discussion concerning the composition of refinements by interpretation is not
straightforward. For \emph{vertical composition} one gets, similarly to what happens in the equational case,
\begin{thm}\label{th:ver2}
Let $\cl$, $\cl'$ and $\cl''$ be $k$, $l$ and $m$-\dlogics\ over $\Sigma$, $\Sigma'$ and
$\Sigma''$ respectively. Let $\tau$ be a $(k,l)$-translation from $\Sigma$ to $\Sigma'$ and
$\rho$ a $(l,m)$-translation from $\Sigma'$ to $\Sigma''$. Suppose that $\cl
\rightharpoondown_\tau \cl'$, $\cl'\rightharpoondown_\rho \cl''$ and $\rho$ interprets
$\cl^\tau$. Then $\cl\rightharpoondown_{\rho\circ\tau} \cl''$
\end{thm}
\begin{proof}
Directly from the fact that $\mathcal{L} \rightharpoondown_\tau \mathcal{L'}$ and $\mathcal{L'}
\rightharpoondown_\rho \mathcal{L''}$ we have that $\Gamma
\vdash_{\mathcal{L}}\bar\varphi$ implies
$\rho(\tau(\Gamma))\vdash_{\mathcal{L''}}
\rho(\tau(\bar\varphi))$ for any $\Gamma\cup \{\bar\varphi\} \subseteq \form{k}{\Sigma}.$

 On the other hand, by
hypothesis, for any $\Gamma,\{\bar\varphi\}\subseteq \form{k}{\Sigma}$,
$$\Gamma\vdash_{\mathcal{L}}\bar\varphi \Leftrightarrow \tau(\Gamma)\vdash_{\mathcal{L}^\tau}
\tau(\bar\varphi)\Leftrightarrow
\rho(\tau(\Gamma))\vdash_{{(\mathcal{L}^\tau)}^\rho}\rho(\tau(\bar\varphi))$$ and hence,
$\rho\circ\tau$ interprets $\mathcal{L}$. Therefore, $\mathcal{L}
\rightharpoondown_{\rho\circ\tau} \mathcal{L''}$.
\end{proof}

On the other hand, \emph{horizontal composition} of refinements via interpretation is still a
topic of current research, which leads us to the conclusions of this paper.

\section{Concluding}\label{sc:conc}

%%%%%%%%%%
%   \input{mmb12CONC.tex}

\subsection{Related work}
The idea of relaxing what counts as a valid refinement of a specification by
replacing \emph{signature morphisms} by \emph{logical interpretations} is, to the best of our
knowledge, new. This piece of research was directly inspired by the first author's work on algebraic logic, where
the notion of \emph{interpretation} plays a fundamental role (see, \eg,
\cite{memoirs,pigozzi,blok,proto}) and occurs in different variants.
Rather than reviewing exhaustively this area, we shall concentrate in what appears to  be
the closest approach,  in the literature, to the notion of an interpretation proposed in the paper
--- that of  \emph{conservative translation} intensively studied by Feitosa and D'Ottaviano
\cite{feitosatese,traducoesconservativas}.
Recall that a conservative translation is a map between deductive system which reflects and preserves logical consequence.
It corresponds thus to an interpretation arising
from a functional translation with $k=l=1$, \ie, between sentential
languages.

The conjunction property, as characterised in the following definition, allows us  to add to the fact that all conservative translations are interpretations (insofar functions are particular cases of multi-functions), its converse, although in a restricted form. Similar properties, also concerning other connectives, have been studied in the framework of the theory of institutions (see [Tar85]).

\begin{defi}\label{conjprop}
A $k$-\dlogic\ $\cl$ over $\Sigma$ has the \emph{conjunction property} if, for any $\{\bar \varphi_i| i\in \textrm{I}\} \subseteq \form{k}{\Sigma}$, for $\textrm{I}$ finite, there exists a $\bar \xi \in \form{k}{\Sigma}$ such that $\{\bar \varphi_i| i\in \textrm{I}\}\dashv\vdash_\cl \bar \xi$. In this case, we denote $\bar \xi$ by $\bigwedge_{\cl} \{\bar \varphi_i | i \in \textrm{I}\}$.

 In the presence of this property, we define the \emph{associated function} of a translation $\tau$ between two \dlogics\ over $\Sigma$ and $\Sigma'$ as follows
 \begin{center}
 \begin{tabular}{cccc}
$f_\tau:$ & $\form{k}{\Sigma}$ & $\rightarrow$ & $\form{l}{\Sigma'}$\\
 & $\bar \varphi$ & $\mapsto$ & $\bigwedge_{\cl'} \tau(\bar\varphi)$.\\
\end{tabular}
\end{center}
\end{defi}

We may now incorporate in the approach proposed in this paper the important tool given in Lemma \ref{th:nsc}:
\begin{lem}
Let $\mathcal{\tau}$ be a self-translation between two $1$-deductive systems $\cl$ and $\cl'$ over the signature $\Sigma$. Then, if $\cl'$ has the conjunction property, $\tau$ is an interpretation iff its associated function is a conservative translation.
\end{lem}
\begin{proof}
First we prove that $f_\tau(\Gamma) \dashv\vdash_{\cl'} \tau(\Gamma)$: since for any $\xi \in f_\tau(\Gamma)$ there is a $\gamma \in \Gamma$ such
that $\xi=\bigwedge_{\cl'}\tau(\gamma)$, we have, by the conjunction property of $\cl'$, that $\tau(\gamma)\dashv\vdash_{\cl'}\xi$ and, by $\textrm{(ii)}$ of Definition \ref{df:klogic}, that $\tau(\Gamma)\vdash_{\cl'} f_\tau(\Gamma)$.	
Analogously, since for each $\gamma \in \Gamma$ there is a $\xi\in f_\tau(\Gamma)$ such $\xi=\bigwedge_{\cl'}\tau(\gamma) \dashv\vdash_{\cl'} \tau(\gamma)$, we have, by $\textrm{(ii)}$ of Definition \ref{df:klogic}, that $f_\tau(\Gamma) \vdash_{\cl'} \tau(\Gamma)$. Hence, for any interpretation $\tau$ and for all $\Gamma, \{\varphi\} \subseteq \form{1}{\Sigma}$, \begin{center}
	$\Gamma \vdash_{\cl}\varphi \Leftrightarrow \tau(\Gamma)\dashv\vdash_{\cl'} f_\tau(\Gamma)\vdash_{\cl'}\tau(\varphi)\dashv\vdash_{\cl'} f_\tau(\varphi)$,
	\end{center}
which implies that $f_\tau$ is a conservative translation. In a similar way, if $f_\tau$ is a conservative translation,
 $\tau$  is an interpretation.
\end{proof}

The connection to conservative translations turns out to be very useful in practice.
The following theorem, which builds on results in \cite{traducoesconservativas},
provides a  sufficient condition for a translation to be an interpretation.
\begin{thm}\label{th:nsc}
Let $\tau$ be a $(k,l)$-translation from
$\Sigma$ to $\Sigma'$, $\cl$
 a $k$-\dlogic\ over $\Sigma$ and $\cl'$ a $l$-\dlogic\ over $\Sigma'$. Suppose that $\tau$ is functional and
 injective. If $\tau(\Cn_\cl(\Gamma))=\Cn_{\cl'}(\tau(\Gamma))$, for every set of formulas
$\Gamma$,
 then $\tau$ interprets $\cl$ in $\cl'$.
\end{thm}
\begin{proof}
From the inclusion $\tau(\Cn_\cl(\Gamma))\subseteq \Cn_{\cl'}(\tau(\Gamma))$ we have that
$\Gamma\vdash_\cl \bar\varphi \Rightarrow  \tau(\Gamma)\vdash_{\cl'} \tau(\bar\varphi)$.
Suppose now that $\tau(\bar\varphi)\in \Cn_{\cl'}(\tau(\Gamma))= \tau(\Cn_\cl(\Gamma))$. Hence
there is a $\bar\psi \in \Cn_\cl(\Gamma)$ such that $\tau (\bar\varphi)=\tau(\bar\psi)$. Since
$\tau$ is injective $\bar\varphi = \bar\psi$, and so, $\bar\varphi \in \Cn_\cl(\Gamma)$, \emph{i.e.},
$\Gamma\vdash_\cl \bar\varphi$.
\end{proof}

The approach to refinement proposed in this paper, in particular when specialised to 2-dimension deductive systems,
  should also be related to the
extensive work of Maibaum, Sadler and Veloso in the 70's and the 80's, as documented, for example, in \cite{Maibaum1,Maibaum2}.
The authors resort to  interpretations between theories and conservative extensions to define a syntactic notion of
refinement according to which
  a specification $SP'$ refines a specification $SP$ if there is an interpretation of $SP'$ into a conservative extension of $SP$. It is
   shown that these refinements can be vertically composed, therefore  entailing stepwise development.  This notion is, however,
   somehow restrictive since it requires all maps to be conservative, whereas in  program development it is usually
   enough to guarantee that  requirements are preserved by the underlying translation. Moreover, in their approach,
    the interpretation edge of a refinement diagram needs to satisfy extra properties.

As related work one should also mention \cite{Fiadeiro:1993:GIT:647322.721361,Vou} where interpretations between theories are studied in the abstract framework of $\pi$-institutions. The first reference is a generalisation of the work of Maibaum and his collaborators, whereas the second one generalises the way  algebraic semantics on sentential logics is dealt with in abstract algebraic logic
to the abstract setting of $\pi$-institutions. Similar developments could arise by considering institutions and their (co-)morphisms  \cite{instituicoes,Dia08,Tar95}. The work of
 Meseguer \cite{meseguer} on \emph{general logics}, in which a  theory of interpretations between logical systems is developed, should
also be mentioned.

Our own approach to refinement by interpretation can be placed between these general works and the original
contribution of Maibaum. Actually, on the one hand, we deal with general $k$-deductive systems therefore subsuming all frameworks above which are based on
 equational  or first order logic (\ie, on specific instances of  $k$-deductive systems). On the other hand, however, our results are
 formulated in terms of a concrete and intuitive notion of a deductive system; their scaling to an abstract, institutional level is still
 to be done.

\subsection{Conclusions and future work}

The paper introduced a new notion of refinement and started  the
development of a corresponding theory of \emph{refinement by interpretation}. The results obtained
and their applications seem promising, in the sense that a number of
useful transformations of (classes of models of) specifications are captured as refinement steps.
In order to clarify  the scope of our results we should point out that
the development in Section 3  can be straightforwardly generalised as
to apply to any Horn fragment of a structural logic $\mathcal{L}$
(\emph{i.e.},  a logic whose axioms have the form $\bigwedge H
\rightarrow c$, where $H \cup \{c\}$ is a subset of the atomic formulas of the logic, with $H$ possibly empty). All one has to do is to represent such a fragment by the natural equivalent deductive system taking the atomic formulas of $\mathcal{L}$ as its set of formulas and a presentation given by the axiomatisation of $\mathcal{L}$.

The generalisation made along sections \ref{sc:gen} and \ref{sc:rbi2}
turns it relevant to the specification meta-level, \ie,
 whenever an implementation step requires a change in the underlying logic. This
 often arises in formal software development with the need for accommodating new requirements
 (as in Example \ref{ex:ordbams})  or when a particular theorem prover, embodying a specific logic,
 is to be used for design validation (as in Example \ref{ex:tool}).
 Our most recent work \cite{studia} is another generalisation effort aiming at reframing this notion of
refinement in a categorical setting based on a characterisation of abstract logics as
coalgebras for the closure system contravariant functor \cite{Pal02} upon the category $\setcat$ of sets  and functions.

To conclude we would like to remark again the 'semantic' perspective from which this work was developed, as
extensively discussed in the Introduction. This entails the need for further research on how refinement by
interpretation, which is entirely based on properties of arbitrary deductive systems, can be smoothly combined
with concrete specification structuring operations. Preliminary work on this topic is  reported in \cite{RMMB11} in which
the emphasis is shifted to  specifications. Current work in this direction
includes the development of a refinement calculus of structured specifications over a $\pi$-institution.

As a general remark we would like to stress again that the approach developed in this paper can be applied to any notion of algebraic specification based on any fixed set of specification structuring combinators, further justifying the relevance of the 'semantic'  perspective adopted in this paper.

A practical limitation of this approach of reducing the specifications to their class of models is that often such classes have infinitely many models.
 However in some cases  subclasses possessing the same theory can be considered instead. For instance, if the class is a finitely generated quasivariety $K$ \cite{Gor98}, and consequently the associated 2-deductive system is finitary with a presentation given by the axiomatisation of the quasivariety, we can replace $K$ by the set of its generators which induces the same deductive system. An example is the class of Boolean Algebras which are generated (actually, as a variety) by the two-element Boolean algebra.

From an application point of view, this  'semantic' approach seems to have its own potentialities which we would like to recall.
Actually, in a number of cases it is relevant, and even mandatory, to start the implementation procedure from a set of models that does
not come from a structured specification. This can be the case when reusing designs  (a recurrent strategy in Engineering)
 or even to express meta-requirements that  cannot be easily accommodated within the classic refinement procedure.
Focussing on  \emph{classes of models}, on the other hand, makes possible to deal  with requirements that cannot
be properly formalised in a specification.
Note this does not entail any  loss of expressivity. The approach proposed in this paper can be tuned to specification refinement in a strict sense: for each specification, one may recursively compute its denotation (a signature and a class of models) and work directly with it.

In general, we believe that this approach has a real application potential, namely to deal
with specifications spanning through different specification logics. Particularly
deserving to be considered, but still requiring further investigation, are
observational logic \cite{BHK03}, hidden logic
\cite{grigore_thesis,conditional_prof,TCS_mart,Martins08} and behavioural logic \cite{Hen97}. In all
of these cases the satisfaction of requirements is discussed up to some particular
satisfaction relation and
their verification is checked with respect to relations obtained by replacing
 strict equality by its underlying notion of satisfaction. In this
context, a semantics based on $k$-structures paves the way
to the unification of all of these approaches. Actually, in all of them, models
consist of algebras whose  $k$-structures are of the form $\langle A,
\Theta\rangle$, where $\Theta$ captures the particular satisfaction relation in each formalism.
 In particular, the
strict models of a (classical) algebraic specification $SP$ consist of algebras $A$ whose
$k$-data structure $\langle A, \Delta_A\rangle$ is a model of  $\models_{\Mod(SP)}$.

Naturally, most of the models of software specifications are not admissible
choices as implementations. Therefore,
 the choice of adequate filters along the implementation process becomes a crucial, although
  not  trivial task. This should be driven by   the system nature (for example, adopting
    observational equality to
deal with objects with encapsulated data). A similar concern is, moreover, shared by other
 general approaches
to formal development, as, for example, \cite{Hen97} in the context of behavioral logic.

A lot of other questions remain to be answered.
One such topic, as mentioned above,  concerns \emph{horizontal composition}
of refinements by interpretation; \emph{vertical} composition raising no special problems as
shown in theorems \ref{th:ver1} and \ref{th:ver2}.
To illustrate the kind of results we are investigating suppose,
for example,  that $\tau$ interprets $SP$ in
$SP'$. The challenge would be to prove that  $\tau$ also interprets an enrichment of $SP$ by axioms in an
appropriate sub-specification of $SP'$.
A closely related issue is the extension of this approach to the level of (structured) specifications.
We believe that this  can be captured  in a somehow
standard way, which will be most relevant in
studying the interplay between  horizontal (\emph{i.e.}, architectural) and vertical (\emph{i.e.}, implementation driven)
levels of specification composition.  For example, the \emph{union} of two specifications will correspond to the union of the
corresponding consequence relations. Actually, a structured specification also defines a class of models and therefore
induces a deductive system.

Another topic to explore
is the equivalence of algebraic specifications up to logical interpretation. As a starting
point, it would be worth to explore the relation $\equiv$ defined as follows:
$SP\equiv SP'$ if there are
interpretations $\tau$ and $\rho$ such that $SP\rightharpoondown_\tau SP'$ and
$SP'\rightharpoondown_\rho SP$. It is not difficult to see that
$SP\models \xi$ implies $SP\models\rho(\tau(\xi))$ and $SP'\models \eta$ implies
$SP'\models\tau(\rho(\eta))$.  More challenging seems to be a stronger equivalence, studied in
the context of equivalence between deductive systems  \cite{CR,memoirs},
which requires interpretations to be \emph{mutually} inverse.

Last but not least, framing refinement by interpretation in the context of  recent works on heterogeneous specification, raises
interesting questions and opens the opportunity for computer-based support.
Actually, classical translations between logics (\eg, modal  into first-order or the latter  into equational logic) are at the basis
of \textsc{Hets} \cite{MML07,Mossakowski05Habil,manualhets}, the heterogeneous specification framework.
To go further in this direction entails the need to regard interpretations
from an institutional point of view \cite{Dia08}, as some sort of
comorphisms, and develop on top of it a calculus of refinements by interpretation.

\section*{Acknowledgements}
{\sloppy The authors express their gratitude to the anonymous referees who raised a number of pertinent questions entailing a
more precise  characterisation of the paper's contributions and a clarification of their scope.
This work was funded by ERDF - European Regional Development Fund through the COMPETE Programme (operational programme for competitiveness) and by National Funds through the FCT (Portuguese Foundation for Science and Technology) within project
 \texttt{FCOMP-01-0124-FEDER-028923} (Nasoni) and the project
 \texttt{PEst-C/MAT/UI4106/2011} with COMPETE number
 \texttt{FCOMP-01-0124-FEDER-022690} (CIDMA - UA). The first author
 also acknowledges the  financial assistance by the projects GetFun,
 reference \texttt{FP7-PEOPLE-2012-IRSES}, and \textsc{Nociones de
   Completud}, reference \texttt{FFI2009-09345} (MICINN - Spain).
 A.\ Madeira was supported by the FCT within the project \texttt{NORTE-01-0124-FEDER-000060}.}

\addcontentsline{toc}{section}{\protect\numberline{}{Bibliography}}

\bibliographystyle{alpha}

%\bibliography{rrffss}

\newcommand{\etalchar}[1]{$^{#1}$}

\end{document}